\documentclass{article}

\usepackage[nonatbib,preprint]{neurips_2020}

\usepackage[utf8]{inputenc} % allow utf-8 input
\usepackage[T1]{fontenc}    % use 8-bit T1 fonts
\usepackage{url}            % simple URL typesetting
\usepackage{booktabs}       % professional-quality tables
\usepackage{amsfonts}       % blackboard math symbols
\usepackage{nicefrac}       % compact symbols for 1/2, etc.
\usepackage{microtype}      % microtypography

\usepackage{mathrsfs}
\usepackage{appendix}

\usepackage{amsmath,amsthm,amssymb}

\usepackage[utf8]{inputenc} % input from keybord is utf8
\usepackage[T1]{fontenc}    % use 8-bit T1 fonts
\usepackage{amsfonts}       % blackboard math symbols
\usepackage{nicefrac}       % compact symbols for 1/2, etc.
\usepackage{microtype}      % microtypography

\usepackage[dvipdfmx]{graphicx}
\usepackage[svgnames]{xcolor} 
\usepackage{subcaption}
\usepackage{wrapfig}

\definecolor{BlueYonder}{RGB}{68, 112, 173}
\definecolor{IshihataBlue}{RGB}{0 135 190}

\usepackage{tabularx}       % tabular with specified width 
\usepackage{booktabs}       % professional-quality tables

\usepackage{algorithm} 
\usepackage{algorithmicx}
\usepackage[noend]{algpseudocode}

\usepackage[square,numbers]{natbib}

\usepackage{thmtools}
\usepackage{thm-restate}
\usepackage[hidelinks=true, hypertexnames=false, colorlinks=true, allcolors=IshihataBlue]{hyperref}       % hyperlinks
\usepackage{url}            % simple URL typesetting
\usepackage{cleveref}
\usepackage{autonum}
\usepackage{mathtools}
\usepackage{tikz}
\usetikzlibrary{positioning, arrows.meta, decorations.pathmorphing}
\usepackage{subcaption}

\usepackage{braket}
\usepackage{multirow}
\usepackage{lipsum}
\allowdisplaybreaks
\usepackage{lscape}

\usepackage{atbegshi,etoolbox}% http://ctan.org/pkg/{lipsum,atbegshi,etoolbox}
\makeatletter
\newcommand{\discardpages}[1]{% \discardpages{<csv list>}
	\xdef\discard@pages{#1}% Store pages to discard
	\AtBeginShipout{% At shipout, decide whether to discard page/not
		\renewcommand*{\do}[1]{% How to handle each page entry in csv list
			\ifnum\value{page}=##1\relax%
			\AtBeginShipoutDiscard% Discard page/not
			\gdef\do####1{}% Do nothing further
			\fi%
		}%
		\expandafter\docsvlist\expandafter{\discard@pages}% Process list of pages to discard
	}%
}
\newif\ifkeeppage
\newcommand{\keeppages}[1]{% \keeppages{<csv list>}
	\xdef\keep@pages{#1}% Store pages to keep
	\AtBeginShipout{% At shipout, decide whether to discard page/not
		\keeppagefalse%
		\renewcommand*{\do}[1]{% How to handle each page entry in csv list
			\ifnum\value{page}=##1\relax%
			\keeppagetrue% Page should be kept
			\gdef\do####1{}% Do nothing further
			\fi%
		}%
		\expandafter\docsvlist\expandafter{\keep@pages}% Process list of pages to keep
		\ifkeeppage\else\AtBeginShipoutDiscard\fi% Discard page/not
	}%
}
\makeatother

\newcommand{\Ical}{\mathcal{I}}

\newcommand{\Orm}{\mathrm{O}}

\renewcommand{\hbar}{\overline{h}}

\newcommand{\R}{\mathbb{R}}

\newcommand{\E}{\mathbb{E}}

\newcommand{\e}{\mathrm{e}}	

\newcommand\ib{\mbox{\boldmath $1$}}

\newcommand\zeros{\mbox{\boldmath $0$}}
\newcommand{\Ib}{{\text{{\boldmath $\mathrm{I}$}}}}

\newtheorem{assump}{Assumption}% Assumption
\newtheorem{defn}{Definition}% Definition

\theoremstyle{definition}
\DeclareMathOperator*{\maximize}{maximize}
\DeclareMathOperator{\subto}{subject\ to}
\DeclareMathOperator*{\argmax}{argmax}

\DeclareMathOperator*{\softmax}{softmax}

\newcommand{\iprod}[1]{\left\langle {#1} \right\rangle}

\newcommand{\ceil}[1]{\left\lceil {#1} \right\rceil}

\newcommand{\diag}{\mathrm{diag}}

\newcommand{\relmid}[1]{\mathrel{#1|}}

\newif\iffigure
\figuretrue

\discardpages{}
\title{Differentiable Greedy Submodular Maximization: Guarantees, Gradient Estimators, and Applications}
\author{%
	Shinsaku Sakaue \\
	Graduate School of Information Sciences and Technology \\
	The University of Tokyo \\
	\texttt{sakaue@mist.i.u-tokyo.ac.jp} \\
}

\begin{document}

\maketitle

\begin{abstract}
Motivated by, e.g., sensitivity analysis and end-to-end learning, the demand for differentiable optimization algorithms has been significantly increasing. In this paper, we establish a theoretically guaranteed versatile framework that makes the greedy algorithm for monotone submodular function maximization differentiable. We smooth the greedy algorithm via randomization, and prove that it almost recovers original approximation guarantees in expectation for the cases of cardinality and $\kappa$-extensible system constrains. We also show how to efficiently compute unbiased gradient estimators of any expected output-dependent quantities. We demonstrate the usefulness of our framework by instantiating it for various applications. 
\end{abstract}

\newcommand{\pvec}{{\text{{\boldmath $\mathrm{p}$}}}}
\newcommand{\qvec}{{\text{{\boldmath $\mathrm{q}$}}}}
\newcommand{\fvec}{{\text{{\boldmath $\mathrm{f}$}}}}
\newcommand{\gvec}{{\text{{\boldmath $\mathrm{g}$}}}}
\newcommand{\xvec}{{\text{{\boldmath $\mathrm{x}$}}}}
\newcommand{\yvec}{{\text{{\boldmath $\mathrm{y}$}}}}
\newcommand{\wvec}{{\text{{\boldmath $\mathrm{w}$}}}}
\newcommand{\thetavec}{{\text{{\boldmath $\mathrm{\theta}$}}}}
\newcommand{\lambdavec}{{\text{{\boldmath $\mathrm{\lambda}$}}}}
\newcommand{\Xb}{{\text{{\boldmath $\mathrm{X}$}}}}

\newcommand{\pveck}[1]{\pvec_{#1}}
\newcommand{\gveck}[1]{\gvec_{#1}}
\newcommand{\Sk}[1]{S_{#1}}

\newcommand{\f}{f}
\newcommand{\fdel}[2]{\f_{{#2}}({#1})}
\newcommand{\fdelp}[3]{\f_{{#2}}({#1}, {#3})}
\newcommand{\F}{F}

\newcommand{\nk}{{n_k}}
\newcommand{\Uk}{{U_k}}

\newcommand{\pk}[1]{p_{#1}}
\newcommand{\svec}{{\text{{\boldmath $\mathrm{s}$}}}}

\newcommand{\Sscr}{\mathscr{S}}
\newcommand{\sgreedy}{{\sc Smoothed Greedy}}

\section{Introduction}\label{sec:introduction}
Submodular function maximization is ubiquitous in practice. 
In many situations including 
budget allocation \citep{alon2012optimizing}, 
data summarization \citep{mirzasoleiman2016fast}, and 
active learning \citep{wei2015submodularity}, 
submodular functions are modeled with parameters. 
Formally, we consider the following parametric submodular function maximization: 
\begin{align}\label{prob:main}
\maximize_{X\subseteq V}\quad \f(X, \thetavec)  \qquad \subto \quad X\in \Ical,
\end{align}
where $V$ is a set of $n$ elements, 
$\f(\cdot, \thetavec): 2^V \to \R$ is a set function with continuous-valued parameter vector $\thetavec\in \Theta$, 
and $\Ical\subseteq 2^V$ is a set family consisting of all feasible solutions. 
As is often the case, 
we assume $\f(\cdot, \thetavec)$ to be normalized, monotone, and submodular for any $\thetavec\in\Theta$ (see, \Cref{subsec:background}).

Once $\thetavec$ is fixed, 
we often apply the greedy algorithm to \eqref{prob:main} since it has strong theoretical guarantees \citep{nemhauser1978analysis,fisher1978analysis} and delivers high empirical performances. 
However, if $\thetavec$ largely deviates from unknown true $\hat \thetavec$, 
the greedy algorithm may return a poor solution to the problem of maximizing $\f(\cdot, \hat \thetavec)$. 
This motivates us to study how changes in $\thetavec$ values affect outputs of the greedy algorithm. 
Furthermore, it is desirable if we can learn $\thetavec$ from data so that the greedy algorithm can achieve high $\f(\cdot, \hat \thetavec)$ values.

A major approach to studying such subjects is to differentiate outputs of algorithms w.r.t. parameter $\thetavec$. Regarding continuous optimization algorithms, this approach has been widely studied in the field of sensitivity analysis \citep{rockafellar1998variational,gal2012advances}, 
and it is used by recent decision-focused (or end-to-end) learning methods \citep{donti2017task,wilder2019melding}, which learn to predict $\thetavec$ based on outputs of optimization algorithms. 
When it comes to the greedy algorithm for submodular maximization, 
however, its outputs are not differentiable since continuous changes in $\thetavec$ cause discrete changes in outputs.
Therefore, for using the well-established methods based on derivatives of outputs, we must employ some kind of smoothing technique.  

\citet{tschiatschek2018differentiable} 
opened the field of differentiable submodular maximization; 
they proposed greedy-based differentiable learning methods for 
monotone and non-monotone submodular functions.   
Their algorithm for monotone objective functions was obtained by replacing non-differentiable $\argmax$ with differentiable $\softmax$. 
Since then, this field has been attracting increasing attention; 
 %differentiable greedy submodular maximization has been an attractive field of study; 
another $\softmax$-based algorithm
that forms a neural network (NN) \citep{powers2018differentiable} and 
applications \citep{kalyan2019trainable,peyrard2019principled} 
have been studied. 
However, this field is still in its infancy 
and the following problems remain open:   

{\centering 
	\it Can we smooth the greedy algorithm without losing its theoretical guarantees? 
	\par}
{\centering 
	\it Can we develop application-agnostic efficient methods for computing derivatives?
	\par}

The first problem is important since, 
without the guarantees, 
we cannot ensure that the differentiation-based methods work well.
The existing studies \citep{tschiatschek2018differentiable,powers2018differentiable} state that the ($1-1/\e$)-approximation 
for the cardinality constrained case is obtained if the temperature of $\softmax$ is zero 
(i.e., equal to $\argmax$). 
This, however, provides no theoretical guarantees for the smoothed differentiable algorithms. 
%; without them, we cannot ensure that the differentiation-based methods work well.
%

As regards the second problem, 
the existing methods \citep{tschiatschek2018differentiable,powers2018differentiable} focus on differentiating some functions defined with subsets $X_1, X_2, \ldots \subseteq V$ given as training data. 
This restricts the scope of application; 
for example, we cannot use them for sensitivity analysis (see, \Cref{sec:comparison} for details). 
The computation cost also matters when developing differentiation methods; in \citep{tschiatschek2018differentiable}, 
a heuristic approximation method is used since the exact computation of derivatives generally incurs exponential costs in $n$.  

{\bf Our contribution} is a theoretically guaranteed versatile framework that resolves the two problems, thus greatly advancing the field of differentiable submodular maximization. 
As shown in \Cref{sec:stochastic_greedy}, 
our framework also works with the stochastic greedy algorithm \cite{mirzasoleiman2015lazier}. 
%, which is more efficient than the greedy algorithm 
%and useful when addressing large-scale instances. 
%a faster variant of the greedy algorithm applicable to the cardinality constrained case. 
%This is useful, e.g., when addressing large-scale decision-focused learning instances.
Below we describe the details.

\begin{description}
	\item[{\bf S{\scriptsize MOOTHED} G{\scriptsize REEDY}}] 
	We develop {\sgreedy} by stochastically perturbing $\argmax$; 
	this generalizes the existing algorithms  \citep{tschiatschek2018differentiable,powers2018differentiable}. 
	We prove that the perturbation does not spoil the original guarantees: 
	almost ($1-1/\e$)- and $\frac{1}{\kappa +1}$-approximation guarantees are achieved in expectation 
	for the cases of cardinality and $\kappa$-extensible system constraints, respectively, where 
	a subtractive term depending on the perturbation strength affects the guarantees.
	\item[{\bf Gradient estimation}]
	Due to the perturbation, we can differentiate expected outputs of {\sgreedy}; 
	the computation cost is, however, exponential in $n$ as with \citep{tschiatschek2018differentiable}. 
	To circumvent this, we show how to compute unbiased gradient estimators of any expected output-dependent quantities by sampling {\sgreedy} outputs.     
	This enables us to efficiently estimate derivatives of, e.g., expected objective values and the probability that  each $v\in V$ is chosen. 
	\item[{\bf Applications}] 
	We demonstrate that our framework can serve as a bridge between the greedy algorithm and differentiation-based methods in many applications. When used for sensitivity analysis, it elucidates how outputs of {\sgreedy} can be affected by changes in $\thetavec$ values. Results of decision-focused learning experiments suggest that our greedy-based approach can be a simple and effective alternative to a recent continuous relaxation method \citep{wilder2019melding}. 
\end{description}

%In \Cref{sec:stochastic_greedy}, 
%we show that our framework also works with the stochastic greedy algorithm~\cite{mirzasoleiman2015lazier}, which is a faster variant of the greedy algorithm applicable to the cardinality constrained case. 
%This is useful, e.g., when addressing large-scale decision-focused learning instances.

\subsection{Related work}\label{subsec:related}
\citet{nemhauser1978analysis} proved the ($1-1/\e$)-approximation guarantee of the greedy algorithm for 
the cardinality constrained case, and this result is known to be optimal \citep{nemhauser1978best, feige1998threshold}. 
\citet{fisher1978analysis} proved that the greedy algorithm achieves the $\frac{1}{\kappa+1}$-approximation 
if $(V, \Ical)$ is an intersection of $\kappa$ matroids; 
later, this result was extended to the class of $\kappa$-systems \citep{calinescu2011maximizing}, 
which includes $\kappa$-extensible systems.

%\citet{calinescu2011maximizing} showed that this result holds 
%if $(V, \Ical)$ belongs to a more general class called the , 
%which includes . 

%Parameterized submodular functions appear in many applications including 
%budget allocation \citep{alon2012optimizing}, 
%data summarization \citep{mirzasoleiman2016fast}, and 
%active learning \citep{wei2015submodularity}. 
%\citet{dolhansky2016deep} developed 
%a flexible family of NN-based parameterized submodular functions, 
%called deep submodular functions.  

Differentiable greedy submodular maximization is studied in \citep{tschiatschek2018differentiable,powers2018differentiable}. 
Our work is different from them in terms of 
theoretical guarantees, differentiation methods, and problem settings 
as explained above (see, also \Cref{sec:comparison}). 
The closest to our result is perhaps that of the continuous relaxation method \citep{wilder2019melding}. 
Specifically, they use the multilinear extension \citep{calinescu2011maximizing} of $\f(\cdot, \thetavec)$ 
and differentiate its
local optimum computed with the stochastic gradient ascent method (SGA) \citep{hassani2017gradient}, 
which achieves a $1/2$-approximation.  
Their method can be used for matroid constraints, but their analysis focuses on the cardinality constrained case. 
Compared with this, our method is 
advantageous in terms of approximation ratios and 
empirical performances (see, \Cref{subsec:exp_decision}).
Note that our method is also different from 
sampling-based methods for leaning 
submodular functions (e.g., \citep{balcan2011learning,rosenfeld2018learning}).

Differentiable end-to-end learning has been studied in many other settings: 
submodular minimization \citep{djolonga2017differentiable}, 
quadratic programming \citep{amos2017optnet}, 
mixed integer programming \citep{ferber2020mipaal}, 
optimization on graphs \citep{wilder2019end}, 
combinatorial linear optimization \citep{pogancic2020differentiation}, 
satisfiability (SAT) instances \citep{wang2019satnet}, 
and ranking/sorting \citep{cuturi2019differentiable}.
%and 
%satisfiability (SAT) instances \citep{wang2019satnet}. 
%
%Probably the first study that developed 
%a differentiable end-to-end learning method is that of \citet{bengio1997using}, 
%which focuses on the financial application. 

Perturbation-based smoothing is used for, e.g., online learning \citep{abernethy2016perturbation}, linear contextual bandit \citep{kannan2018smoothed}, linear optimization \citep{berthet2020learning}, 
and sampling from discrete distributions \citep{gumbel1954statistical,jang2017categorical,maddison2017concrete}, but it has not been theoretically studied for smoothing the greedy algorithm for monotone submodular maximization. 
%
%
%
%The design of  {\sgreedy} is related to the link between perturbation and regularization (see, e.g., \citep{abernethy2016perturbation}). 

\subsection{Notation and definition}\label{subsec:background}
For any set function $\f:2^V\to\R$, we define $\fdel{Y}{X} \coloneqq \f(X \cup Y) - \f(X)$. 
We say $\f$ is 
normalized if $\f(\emptyset) = 0$, 
monotone if $X\subseteq Y$ implies $\f(X) \le \f(Y)$, 
and 
submodular if $\fdel{v}{X} \ge \fdel{v}{Y}$ 
for all $X\subseteq Y$ and $v \notin Y$. 
In this paper, 
we assume the objective function, $\f(\cdot, \thetavec)$, to be 
normalized, monotone, and submodular for any $\thetavec\in\Theta$. 
Note that this is the case with many set functions, e.g., 
weighted coverage functions with non-negative weights $\thetavec$, 
probabilistic coverage functions with probabilities $\thetavec$, 
and deep submodular functions \citep{dolhansky2016deep} 
with non-negative linear-layer parameters $\thetavec$.

We say $(V, \Ical)$ is a $\kappa$-extensible system \citep{mestre2006greedy} if
the following three conditions hold: 
(i) $\emptyset\in\Ical$, 
(ii) $X\subseteq Y \in \Ical$ implies $X \in \Ical$, 
and
(iii) for all $X\in \Ical$ and $v\notin X$ such that $X \cup \{v\} \in \Ical$, 
and for every $Y\supseteq X$ such that $Y\in \Ical$, 
there exists $Z \subseteq Y\backslash X$ that satisfies $|Z|\le \kappa$ and 
$Y\backslash Z \cup \{v\}\in \Ical$. 
As shown in \citep{mestre2006greedy}, 
$(V, \Ical)$ is a matroid iff it is a $1$-extensible system, 
which includes the cardinality constrained case, 
and the intersection of $\kappa$ matroids defined on a common ground set always forms a $\kappa$-extensible system. 
We say $X\in \Ical$ is maximal if no $Y\in\Ical$ strictly includes $X$. 
We define $K \coloneqq \max_{X\in \Ical} |X|$, which is so-called the rank of $(V, \Ical)$.

For any positive integer $n$, 
we let $\zeros_n$ and $\ib_n$ be $n$-dimensional all-zero and all-one vectors, respectively. 
For any finite set $V$ and $S\subseteq V$, 
we let $\ib_S \in \R^{|V|}$ denote 
the indicator vector of $S$; i.e., 
the entries corresponding to $S$ are $1$ and the others are $0$. 
Given any scalar- or vector-valued differentiable function $\fvec:\R^n \to \R^m$, 
$\nabla_\xvec \fvec(\xvec) \in \R^{m\times n}$ denotes its 
gradient or Jacobian, respectively. 

\section{Smoothed greedy algorithm}\label{sec:greedy}
We present {\sgreedy} (\Cref{alg:pgreedy}) and prove its approximation guarantees. 
In this section, we take parameter $\thetavec\in \Theta$ to be fixed arbitrarily.

\begin{algorithm}[tb]
	\caption{\sgreedy}
	\label{alg:pgreedy}
	\begin{algorithmic}[1]
		\State $S \gets \emptyset$
		\For{$k = 1,2\dots$}
		\State $\Uk = \{u_1,\dots,u_{\nk}\} \gets \{v \notin S \relmid{} S\cup\{v\} \in\Ical\}$
		\State $\gveck{k}(\thetavec) = 
		(g_k(u_1, \thetavec),\dots, g_k(u_\nk, \thetavec))
		\gets (\fdelp{u_1}{S}{\thetavec},\dots,\fdelp{u_{\nk}}{S}{\thetavec})$ \label{step:g}
		\State $\pveck{k}(\thetavec) = 
		(\pk{k}(u_1, \thetavec),\dots, \pk{k}(u_\nk, \thetavec))
		\gets \argmax_{\pvec\in\Delta^{\nk}}\{\iprod{\gveck{k}(\thetavec), \pvec} - \Omega_k(\pvec)   \}$
		\State $s_k \gets u\in \Uk$ with probability $\pk{k}(u, \thetavec)$ \label{step:choose}
		\State $S\gets S\cup\{s_k\}$ 
		\If{$S$ is maximal} \Return $S$
		\EndIf
		\EndFor        
	\end{algorithmic}
\end{algorithm}

We explain the details of \Cref{alg:pgreedy}. 
In the $k$-th iteration, we compute marginal gain  
$\fdelp{u}{S}{\thetavec}$ for every addable element $u\in \Uk\coloneqq \{v\notin S \relmid{} S\cup \{v\} \in\Ical \}$; 
we define $\nk\coloneqq |\Uk|$ and 
index the elements in $\Uk$ as $u_1, \dots, u_\nk$. 
Let $\gveck{k}(\thetavec)\in\R^\nk$ denote the marginal gain vector.  
We then compute
\begin{align}\label{eq:solveomega}
\pveck{k}(\thetavec) = \argmax_{\pvec\in\Delta^\nk}\{\iprod{\gveck{k}(\thetavec), \pvec} - \Omega_k(\pvec)   \},
\end{align}
where $\Delta^\nk\coloneqq \{ \xvec \in \R^{\nk}  \relmid{} \xvec\ge \zeros_\nk, \iprod{\xvec, \ib_\nk}=1  \}$ is 
the ($\nk-1$)-dimensional probability simplex 
and   
$\Omega_k:\R^{\nk}\to \R$ is a strictly convex function; 
we call $\Omega_k$ a regularization function.  
Note that 
the strict convexity implies the uniqueness of $\pveck{k}(\thetavec)$.\footnote{
	Note that $\pveck{k}(\thetavec)$ depends on the past $k-1$ steps, 
	which we do not write explicitly for simplicity.
%	A remark on the notation, $\pveck{k}(\thetavec)$:  
%	since $\pveck{k}(\thetavec)$ depends on $\fdelp{u}{S}{\thetavec}$ where $S = \{s_1,\dots, s_{k-1}\}$, 
%	it would be better to explicitly write the dependence on the past $k-1$ steps as, e.g., $\pveck{k}(\thetavec, s_1,\dots, s_{k-1})$. 
%	For simplicity, however, we omit $s_1,\dots, s_{k-1}$ in what follows; 
%	the dependence will be clear from the context.  
}  
We then choose an element, $u\in \Uk$, with probability $\pk{k}(u, \thetavec)$; 
let $s_k$ denote the chosen element.  
The above procedure can be seen as a stochastically perturbed version of $\argmax$; 
without $\Omega_k$, we have $s_k \in \argmax_{u \in \Uk} \fdelp{u}{S}{\thetavec}$.

We then study theoretical guarantees of {\sgreedy} 
(we present all proofs in \Cref{asec:proofs}).
Let $\delta\ge0$ be a constant that satisfies
$\delta\ge \Omega_k(\pvec) - \Omega_k(\qvec)$ 
for all $k=1,\dots, |S|$ and $\pvec, \qvec\in \Delta^{\nk}$. 
We will shortly see that 
smaller $\delta$ values yield better guarantees;  
we present examples of $\Omega_k$ and their $\delta$ values at the end of this section.

%As is often the case with the analysis of the greedy algorithm, 
As is often done, we begin by lower bounding the marginal gain. 
The following lemma elucidates the effect of $\delta$ 
and plays a key role when proving the subsequent theorems. 
\begin{restatable}{lem}{marginal}\label{lem:fdel}
	In any $k$-th step, conditioned on the ($k-1$)-th step 
	(i.e., $S = \{ s_1,\dots, s_{k-1} \}$ is arbitrarily fixed), we have 
	$\E[\fdelp{s_k}{S}{\thetavec}] \ge \fdelp{u}{S}{\thetavec} - \delta$
	for any $ u\in \Uk$.
\end{restatable}

Let $S$  and $O$ be an output of \Cref{alg:pgreedy}
and 
a maximal optimal solution to problem \eqref{prob:main}, respectively. 
In the cardinality constrained case, we can obtain the following guarantee. 
We also show in \Cref{thm:ssgreedy} (\Cref{subsec:stoc_approximation_guarantee}) that the faster stochastic variant \citep{mirzasoleiman2015lazier} can achieve a similar approximation guarantee. 

\begin{restatable}{thm}{cardinality}\label{thm:cardinality} 
	If $\Ical = \{ X \subseteq V \relmid{} |X| \le K  \}$, we have 
	$\E[\f(S, \thetavec)] \ge (1-1/\e) \f(O,\thetavec) - \delta K$. 
\end{restatable}
For the more general case of $\kappa$-extensible systems, we can prove the following theorem. 
\begin{restatable}{thm}{extensible}\label{thm:extensible}
	If $(V, \Ical)$ is a $\kappa$-extensible system with rank $K$, we have 
	$\E[\f(S, \thetavec)] \ge \frac{1}{\kappa+1} \f(O,\thetavec) - \delta K$. 
\end{restatable}
\begin{proof}[Proof sketch of \Cref{thm:extensible}]
	First, we briefly review the proof for the standard greedy algorithm \citep{calinescu2011maximizing}. 
	For a series of subsets $\emptyset = S_0 \subseteq S_1 \subseteq \dots \subseteq S_{|S|} = S$ obtained in $|S|$ steps of the greedy algorithm, 
	we construct a series of subsets $O = O_0, O_1\dots,O_{|S|} = S$ that satisfies $S_i \subseteq O_i \in\Ical$ 
	and 
	$ \kappa \cdot (\f(S_{i}, \thetavec)- \f(S_{i-1}, \thetavec)) \ge \f(O_{i-1}, \thetavec) - \f(O_i, \thetavec) $ for $i=1,\dots,|S|$.
	The $\frac{1}{\kappa +1}$-approximation is obtained by summing both sides for $i=1,\dots,|S|$. 
	Our proof extends this analysis to the randomized {\sgreedy}. 
	We construct $O_0, O_1\dots$ 
	for each realization of the randomness, and prove
	\[
	\kappa \cdot (\E[\f(S_{i}, \thetavec)] - \E[\f(S_{i-1}, \thetavec)] + \delta) \ge \E[\f(O_{i-1}, \thetavec)] - \E[\f(O_i, \thetavec)]
	\]	 
	for $i=1,\dots, K$ by using \Cref{lem:fdel}, where we must carefully deal with the fact 
	that $|S| < K$ may occur in some realizations. 
	By summing both sides for $i=1,\dots, K$, 
	we obtain \Cref{thm:extensible}.    
\end{proof}

Existing guarantees 
\citep{tschiatschek2018differentiable,powers2018differentiable} 
only consider the case of \Cref{thm:cardinality} with $\delta = 0$.  
Therefore, our results bring significant progress in theoretically understanding differentiable submodular maximization. 

Below we showcase two examples of regularization function $\Omega_k$: entropy and quadratic functions.  
We can also use other strictly convex functions, e.g., a convex combination of the two functions. 
Note that 
when designing $\Omega_k$, an additional differentiability condition (see, \Cref{assump:diff_p} in \Cref{sec:gradient}) must be satisfied for making expected outputs of {\sgreedy} differentiable.

\paragraph{Entropy function}
Let 
$\Omega_k(\pvec) = \epsilon \sum_{i=1}^{\nk} p(u_i) \ln p(u_i)$, 
where 
$p(u_i)$ is the $i$-th entry of $\pvec\in [0,1]^\nk$ 
and $\epsilon > 0$ is an arbitrary constant. 
In this case, we have 
$\delta = \epsilon \ln \nk$, 
and thus we can make the $\delta$ value arbitrarily small 
by controlling the $\epsilon$ value. 
Moreover, Steps \ref{step:g} to \ref{step:choose} can be efficiently performed 
via softmax sampling as with \citep{tschiatschek2018differentiable,powers2018differentiable}; 
i.e., $p_k(u, \thetavec) \propto \exp(\fdelp{u}{S}{\thetavec} / \epsilon)$ 
(see, \Cref{subsec:entropy}).
\paragraph{Quadratic function}
We can use strongly convex quadratic functions as $\Omega_k$. 
To be specific, if we let $\Omega_k(\pvec) = \epsilon \| \pvec \|_2^2$, 
then $\delta = \epsilon(1 - 1/\nk) \le \epsilon$. 
In this case, we need to solve quadratic programming (QP) problems for $k=1,2,\dots$. 
If we use the same $\Omega_k$ for every $k$, 
preconditioning (e.g., decomposition of Hessian matrices) is effective. 
We can also use an efficient batch QP solver~\citep{amos2017optnet}.

As above, the $\delta$ value is typically controllable, 
which we can use as a hyper-parameter that balances the trade-off between the approximation guarantees and smoothness.  
How to set the $\delta$ value should be discussed depending on applications (see, \Cref{sec:application}).  

\newcommand{\q}{q}
\newcommand{\Q}{Q}
\newcommand{\model}{m}
\section{Gradient estimation}\label{sec:gradient}
We show how to differentiate outputs of {\sgreedy} w.r.t. $\thetavec$; 
the derivative computation method presented in this section 
also works with the stochastic version \citep{mirzasoleiman2015lazier} 
of {\sgreedy} (see, \Cref{subsec:stoc_gradient}). 
In this section, we assume the following two differentiability conditions to hold: 
\begin{assump}\label{assump:diff_f}
	For any $X\subseteq V$, 
	we assume $f(X, \thetavec)$ to be differentiable w.r.t. $\thetavec$.
\end{assump}

\begin{assump}\label{assump:diff_p}
	For any $\thetavec\in\Theta$ and $k\in\{1,\dots, K\}$, 
	let $\pvec_k(\gvec_k)$ be the maximizer, $\pvec_k(\thetavec)$, 
	in \eqref{eq:solveomega} regarded as a function of $\gvec_k(\thetavec)$. 
	We assume $\pvec_k(\gvec_k)$ to be differentiable w.r.t. $\gvec_k$. 
\end{assump}

\Cref{assump:diff_f} is inevitable; the existing studies \citep{tschiatschek2018differentiable,powers2018differentiable,wilder2019melding} 
are also based on this condition. 
Examples of functions satisfying \Cref{assump:diff_f} include 
weighted coverage functions (w.r.t. weights of covered vertices), 
probabilistic coverage functions \citep{wilder2019melding},
and deep submodular functions with smooth activation functions \citep{dolhansky2016deep}. 
At the end of this section, 
we discuss what occurs if \Cref{assump:diff_f} fails to hold and possible remedies for addressing such cases in practice. 

\Cref{assump:diff_p} can be satisfied by appropriately designing $\Omega_k$. 
For example, if $\Omega_k$ is the entropy function,  
the $i$-th entry of $\pvec_k(\gvec_k)$ is $\exp( \epsilon^{-1}  g_k(u_i, \thetavec) ) / \sum_{u\in\Uk} \exp( \epsilon^{-1} g_k(u, \thetavec))$, 
which is differentiable w.r.t. $\gvec_k$. 
%If we use strongly convex quadratic functions as $\Omega_k$, 
%the differentiability condition holds if the strict complementarity is satisfied at $\pvec_k(\thetavec)$ (see, \citep{amos2017optnet}). 
In \Cref{subsec:sufficient}, 
we present a sufficient condition for $\Omega_k$ to satisfy \Cref{assump:diff_p}.

We then introduce the probability distribution of {\sgreedy} outputs.\footnote{
	Although a similar notion is considered in \citep{tschiatschek2018differentiable}, 
	our way of using it is completely different (see, \Cref{sec:comparison}).
}

\begin{defn}[Output distribution]
	Let $\Sscr_{\le K}$ denote the set of all sequences consisting of at most $K$ elements in $V$. 
	For any fixed $\thetavec\in \Theta$, we define  
	$p(\thetavec):\Sscr_{\le K}\to[0,1]$ as the probability distribution function 
	of {\sgreedy} outputs, i.e., $S\sim p(\thetavec)$, 
	which we refer to as the output distribution. 
	We use $p(S, \thetavec)\in[0,1]$ to denote the probability that $S\in\Sscr_{\le K}$ is returned by {\sgreedy}.     
	Specifically, 
	for sequence $S = (s_1, \dots, s_{|S|}) \in \Sscr_{\le K}$ constructed by {\sgreedy}, 
	we let $p(S, \thetavec) = \prod_{k=1}^{|S|} \pk{k}(s_k, \thetavec)$, 
	where $\pk{k}(s_k, \thetavec)$ is the entry of $\pvec_k(\thetavec)$
	corresponding to $s_k\in \Uk$. 
\end{defn}

%For example, if $V=\{1,2,3\}$, 
%%then $\Sscr_{\le 2} =\{ (), (1), (2), (3), (1,2), (2,3), (1,3), (2,1),  (3,2), (3,1) \}$. 
%then $\Sscr_{\le 2}$ consists of 
%$()$, 
%$(1) $, 
%$(2) $, 
%$(3) $, 
%$(1,2)$, 
%$(2,3)$, 
%$(1,3)$, 
%$(2,1)$, 
%$(3,2)$, and 
%$(3,1)$. 
%Note that 
%for a sequence $S = (s_1, \dots, s_{|S|})$ constructed by {\sgreedy}, 
%we have $p(S, \thetavec) = \prod_{k=1}^{|S|} \pk{k}(s_k, \thetavec)$, 
%where $\pk{k}(s_k, \thetavec)$ is the entry of $\pvec_k(\thetavec)$
%corresponding to $s_k\in \Uk$. 
%%

We present our derivative computation method. 
Let $\Q(S)$ be any scalar- or vector-valued quantity; see,  \Cref{sec:application} for examples of $Q(S)$. 
We aim to compute 
$\nabla_\thetavec \E_{S\sim p(\thetavec)}[\Q(S)] = 
\Sigma_{S\in\Sscr_{\le K}} \Q(S)  \nabla_\thetavec p(S,\thetavec)$. 
Since the size of $\Sscr_{\le K}$ is exponential in $K=\mathrm{O}(n)$, 
we usually cannot compute the exact derivative in practice. 
Therefore, we instead use the following unbiased estimator of the derivative:\footnote{
	The above type of estimator is called the score-function gradient estimator \citep{rubinstein1996score} 
	(a.k.a. the likelihood estimator \citep{glynn1990likelihood} and REINFORCE \citep{williams1992simple}). 
	Other than this, 
	there are several major gradient estimators (see, \citep{mohamed2019monte}). 
	In \Cref{sec:other_grad}, we discuss why it is difficult to use those gradient estimators in our setting. 
} 
\begin{restatable}{prop}{gradient}\label{gradient} 
	An unbiased estimator of $\nabla_\thetavec \E_{S\sim p(\thetavec)}[\Q(S)]$ can be obtained by sampling $N$ outputs of {\sgreedy} as follows: 
%	$
%	\frac{1}{N} \sum_{j = 1}^{N} \Q(S_j) \nabla_\thetavec \ln p(S_j,\thetavec)  
%	$,
%	where $S_j = (s_{1}, \dots, s_{|S_j|}) \sim p(\thetavec)$. 	
	\begin{align}\label{eq:estimate}
	\frac{1}{N} \sum_{j = 1}^{N} \Q(S_j) \nabla_\thetavec \ln p(S_j,\thetavec)  
	\quad \text{where} \quad 
	S_j = (s_{1}, \dots, s_{|S_j|}) \sim p(\thetavec). 
	\end{align}
\end{restatable}
\begin{proof}
	We can immediately obtain the result from the following equation:  
	\begin{align}
	\nabla_\thetavec \E_{S\sim p(\thetavec)}[\Q(S)] 
	%&= \Sigma_{S\in\Sscr_{\le K}} \Q(S) \nabla_\thetavec p(S,\thetavec) \\ &
	= \Sigma_{S\in\Sscr_{\le K}} \Q(S) p(S,\thetavec) \nabla_\thetavec \ln p(S,\thetavec) 
	=\E_{S\sim p(\thetavec)}[\Q(S) \nabla_\thetavec \ln p(S,\thetavec)],          
	\end{align} 
	where an unbiased estimator of the RHS can be computed as described in the proposition.	
\end{proof}

The remaining problem is how to compute $\nabla_\thetavec \ln p(S,\thetavec)$ for sampled sequence $S = (s_1, \dots, s_{|S|})$. 
Since we have
$
\nabla_\thetavec \ln p(S,\thetavec)  
= \nabla_\thetavec \ln \prod_{k=1}^{|S|} \pk{k}(s_k, \thetavec)
= \sum_{k=1}^{|S|}  \frac{1}{\pk{k}(s_k, \thetavec)} \nabla_\thetavec \pk{k}(s_k, \thetavec)
$, it suffices to compute 
$\nabla_\thetavec \pk{k}(s_k, \thetavec)$ 
for $k\in \{1,\dots, |S|\}$. 
From Assumptions \ref{assump:diff_f} and \ref{assump:diff_p}, 
we can differentiate $\pveck{k}(\thetavec)$ by using the chain rule as 
$\nabla_\thetavec \pveck{k}(\thetavec)
=
\nabla_{\gvec_k} \pveck{k}(\gvec_k) \cdot \nabla_\thetavec \gvec_k(\thetavec)$, 
and the row corresponding to $s_k\in\Uk$ is equal to 
$\nabla_\thetavec \pk{k}(s_k, \thetavec)$. 
In some cases where 
we can analytically express $\pvec_k(\thetavec)$ as a simple function of $\thetavec$, 
we can directly compute $\nabla_\thetavec \ln p(S,\thetavec)$ 
via efficient automatic differentiation \citep{paszke2017automatic,baydin2018automatic}. 

Regarding the computation complexity, 
if $\Omega_k$ is the entropy function and $\nabla_\thetavec \gvec_k(\thetavec)$ is given, 
we can compute $\nabla_\thetavec \pveck{k}(\thetavec)$ in $\mathrm{O}(\nk \times \dim \Theta)$ time; 
we analyze the complexity in detail in \Cref{sec:regularization}.

%
%
%In the case of, e.g., the entropy regularization functions, once we obtain $\nabla_\thetavec \gvec_k(\thetavec)$, we can compute $\nabla_\thetavec \pveck{k}(\thetavec)$ in $\mathrm{O}(\nk)$ time; we discuss the complexity in detail in \Cref{sec:regularization}. 
%
%In some cases where we can analytically express $\pvec_k(\thetavec)$ as a simple function of $\thetavec$, 
%we can directly compute $\nabla_\thetavec \ln p(S,\thetavec)$ 
%via efficient automatic differentiation \citep{paszke2017automatic,baydin2018automatic}. 

\paragraph{Variance reduction}
The variance of the gradient estimators sometimes becomes excessive, which requires us to sample too many outputs of {\sgreedy}. 
Fortunately, there are various methods for reducing the variance of such Monte Carlo gradient estimators \citep{greensmith2004variance,tucker2017rebar,mohamed2019monte}. 
A simple and popular method is the following baseline correction \citep{williams1992simple}:  
we use $\Q(S) - \beta$ instead of $\Q(S)$, where $\beta$ is some coefficient.  
If $\beta$ is a constant, the estimator remains unbiased since $\E_{S\sim p(\thetavec)}[\nabla_\thetavec \ln p(S,\thetavec)] 
=\nabla_\thetavec \E_{S\sim p(\thetavec)}[1] = 0$.  
By appropriately setting the $\beta$ value, we can reduce the variance. 
In practice, $\beta$ is often set at the running average of $\Q(\cdot)$ values,  
which we use in the experiments (\Cref{sec:experiments}).

\paragraph{Non-differentiable cases}
If \Cref{assump:diff_f} does not hold, i.e., 
$\f(X,\thetavec)$ is not differentiable w.r.t. $\thetavec$, 
the above discussion is not correct since
the chain rule fails to hold \citep{griewank2008evaluating}. 
This issue is common with many machine learning scenarios, 
e.g., training of NNs with ReLU activation functions. 
The current state of affairs is that we disregard this issue since it rarely brings harm in practice. 
Recently, \citet{kakade2018probably} developed a subdifferentiation method for dealing with such non-differentiable cases; this result may enable us to extend the scope of our framework to non-differentiable $\f(X, \thetavec)$.

\section{Applications}\label{sec:application}
Owing to the flexible design of our framework, which accepts any computable $Q(S)$, 
we can use it in various situations. 
We here show how to apply it to sensitivity analysis and decision-focused learning. 
We also present another application related to learning of submodular models 
in \Cref{sec:other_application}.

\subsection{Sensitivity analysis}\label{subsec:sensitivity}
When addressing parametric optimization instances, the sensitivity---how and how much changes in parameter values can affect 
outputs of algorithms---is a major concern, and hence widely studied. 
In continuous optimization settings, 
most sensitivity analysis methods are based on derivatives of outputs
\citep{rockafellar1998variational,gal2012advances,bertsekas2016nonlinear}. 
In contrast, those for combinatorial settings are diverse 
\citep{gusfield1980sensitivity,bertsimas1988probabilistic,ghosh2000sensitivity,varma2019average} 
probably due to the non-differentiability; 
the score-function estimator is also used for 
analyzing the sensitivity of discrete systems (e.g., querying systems) \citep{kleijnen1996optimization}. 
As explained below, 
our gradient estimation method can be used for analyzing the sensitivity of {\sgreedy}, 
which becomes arbitrarily close to the greedy algorithm by letting $\delta$ be sufficiently small. 
This provides, 
to the best of our knowledge, 
the first method for analyzing the sensitivity of the greedy algorithm for submodular maximization.  

We analyze the sensitivity of 
the probability that each $v\in V$ 
is included in an output of {\sgreedy}, 
which can be expressed as 
$\E_{S\sim p(\thetavec)}[\ib_S]
= \Sigma_{S\in\Sscr_{\le K}} \ib_S  p(S,\thetavec)$. 
By using  
our method in \Cref{sec:gradient} with 
$\Q(S) = \ib_S$, 
we can estimate the Jacobian matrix as 
\[
\nabla_\thetavec \E_{S\sim p(\thetavec)}[\ib_S]
\approx
\frac{1}{N} \sum_{j = 1}^{N} \ib_{S_j} \nabla_\thetavec \ln p(S_j,\thetavec).  
\]
Here, given any $\thetavec$, 
the $(v, j)$ entry of the Jacobian matrix represents how and how much the infinitesimal increase in the $j$-th entry of $\thetavec$ 
affects the probability that $v\in V$ is chosen; 
this quantifies the sensitivity of each $v\in V$ to uncertainties in $\thetavec$ values.  
This information 
will be beneficial to practitioners who address tasks involving 
submodular maximization with uncertain parameters; 
for example, 
advertisers who want to know how to reliably promote their products. 
%(we will further discuss this in Broader Impact section). 
In \Cref{subsec:exp_sensitivity}, 
we experimentally demonstrate how this sensitivity analysis method works.

\subsection{Decision-focused learning}\label{subsec:decision}
We consider a situation where $\thetavec$ is computed with some predictive models (e.g., NNs). 
Let $\model(\cdot, \wvec)$ be a predictive model that maps 
some observed feature $\Xb$ to $\thetavec$, 
where $\wvec$ represents model parameters. 
We train $\model(\cdot, \wvec)$ by optimizing $\wvec$ values 
with training datasets $(\Xb_1, \thetavec_1), \dots, (\Xb_M, \thetavec_M)$. 
Given test instance $(\hat \Xb, \hat\thetavec)$, 
where $\hat\thetavec$ is the unknown true parameter, 
the trained model predicts $\thetavec=\model(\hat \Xb, \wvec)$, 
and we obtain solution $S\in \Ical$ (or, make a decision) 
by approximately maximizing $\f(\cdot, \thetavec)$.  
Our utility (decision quality) is measured by $\f(S, \hat \thetavec)$. 
This situation often occurs in real-world scenarios, e.g., budget allocation, diverse recommendation, and viral marketing (see, \citep{wilder2019melding}). 
For example, in the case of viral marketing on a social network, 
$\thetavec$ represents link probabilities, 
which we predict with $\model(\cdot, \wvec)$ for observed feature $\Xb$. 
A decision is a node subset $S$,
which we activate to maximize the influence.
Our utility is the influence spread $\f(S, \hat\thetavec)$, 
where $\hat\thetavec$ represents unknown true link probabilities.

With the decision-focused learning approach \citep{donti2017task,wilder2019melding}, 
we train predictive models in an attempt to maximize the decision quality, 
$\f(S, \hat \thetavec)$.  
This approach is empirically more effective for the above situation, 
which involves both prediction and optimization, 
than the standard two-stage approach.\footnote{
The two-stage approach deals with prediction and optimization separately; 
i.e., we train predictive models with some loss functions defined in advance and then make decisions by approximately maximizing $\f(\cdot, \thetavec)$. 
} 
By combining our framework with the decision-focused approach, 
we can train predictive models with first-order methods so that 
{\sgreedy} achieves high expected objective values. 

Below we detail how to train predictive models with 
our framework and stochastic first-order methods. 
We consider minimizing an empirical loss function defined as  
$- \frac{1}{M} \sum_{i=1}^M \E_{S\sim p(\model(\Xb_i, \wvec))}[\f(S, \thetavec_i)]$, 
where $p(\cdot)$ is the output distribution. 
In each iteration, we sample a training dataset, $(\Xb_i, \thetavec_i)$, 
and compute $\thetavec = \model(\Xb_i, \wvec)$ 
with the current $\wvec$ values. 
We then perform $N$ trials of {\sgreedy} 
to estimate the current loss function value, $-\E_{S\sim p(\thetavec)}[\f(S, \thetavec_i)]$. 
Next, we estimate the gradient by using our method with $Q(S) = \f(S, \thetavec_i)$. 
More precisely, for each $j$-th trial of {\sgreedy}, 
we compute $\nabla_\thetavec \ln p(S_j, \thetavec)$ as explained in \Cref{sec:gradient} and estimate the gradient as follows:\footnote{
	The chain rule, 
	$\nabla_\wvec \ln p(S_j, \model(\Xb_i, \wvec)) = 
	\nabla_\thetavec \ln p(S_j, \thetavec) |_{\thetavec=\model(\Xb_i, \wvec)} \cdot \nabla_\wvec \model(\Xb_i, \wvec)$, 	
	fails to hold if $\model(\cdot, \wvec)$ is not differentiable. 
	This issue is essentially the same as what we discussed in the last paragraph in \Cref{sec:gradient}, which we can usually disregard in practice.  
} 
\begin{align}
-\nabla_\wvec\E_{S\sim p(\model(\Xb_i, \wvec))}[\f(S, \thetavec_i)]
%&=
%-\nabla_\wvec\sum_{S\in\Sscr_{\le K}}\f(S, \thetavec_i) p(S, \model(\Xb_i, \wvec))\\
%&=
%-\sum_{S\in\Sscr_{\le K}}\f(S, \thetavec_i) \nabla_\thetavec p(S, \thetavec) |_{\thetavec=\model(\Xb_i, \wvec)} \cdot \nabla_\wvec \model(\Xb_i, \wvec) \\
%&
\approx
-\frac{1}{N} \sum_{j = 1}^{N} \f(S_j, \thetavec_i) \nabla_\thetavec \ln p(S_j, \thetavec) |_{\thetavec=\model(\Xb_i, \wvec)} \cdot \nabla_\wvec \model(\Xb_i, \wvec).  
\end{align}
%where $\nabla_\wvec \model(\Xb_i, \wvec)$ is obtained by differentiating predictive model $\model$. 
Note that the $N$ trials of {\sgreedy}, 
as well as the computation of $\nabla_\thetavec \ln p(S_j, \thetavec)$, can be performed in parallel. 
We then update $\wvec$ with the above gradient estimator. 
When using mini-batch updates, 
we accumulate the loss values and gradient estimators 
over datasets in a mini-batch, 
and then update $\wvec$.  
Experiments in \Cref{subsec:exp_decision} confirm the practical effectiveness of the above method.

In this setting, the $\delta$ value of $\Omega_k$ should not be too small. 
This is because in early stages of training, {\sgreedy} with small $\delta$ values may overfit to outputs of the predictive model that is not well trained. 
It can be effective to control the $\delta$ values depending on the stages of training.

\section{Experiments}\label{sec:experiments}
We evaluate our method with sensitivity analysis and decision-focused learning instances. 
As a regularization function of \Cref{alg:pgreedy}, we use 
the entropy function
%$\Omega_k(\pvec) = \epsilon \sum_{u\in \Uk} p(u) \ln p(u)$ 
with $\epsilon = 0.2$.  
All experiments are performed on a $64$-bit macOS machine with $1.6$GHz Intel Core i$5$ CPUs and $16$GB RAMs.

We use bipartite influence maximization instances described as follows. 
Let $V$ and $T$ be sets of items and targets, respectively, 
and $\thetavec\in [0,1]^{V \times T}$ be link probabilities.  
We aim to maximize the expected number of influenced targets, 
$\f(X, \thetavec) = \sum_{t\in T} \left(1 -  \prod_{v\in X}  (1 - \theta_{v, t})  \right)$, by choosing up to $K$ items.  

In \Cref{subsec:exp_blackbox}, 
we perform experiments with another setting, where we consider learning deep submodular functions under a partition matroid constraint. 

\begin{figure}[tb]
	\centering
	\begin{minipage}[t]{.2\textwidth}
		\includegraphics[width=1.0\textwidth]{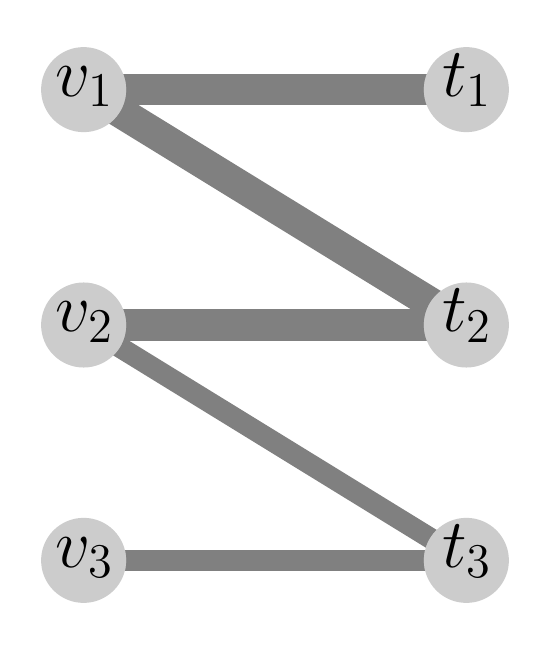}
		\subcaption{$\thetavec$ values}
		\label{fig:prob}
	\end{minipage}
	\begin{minipage}[t]{.2\textwidth}
		\includegraphics[width=1.0\textwidth]{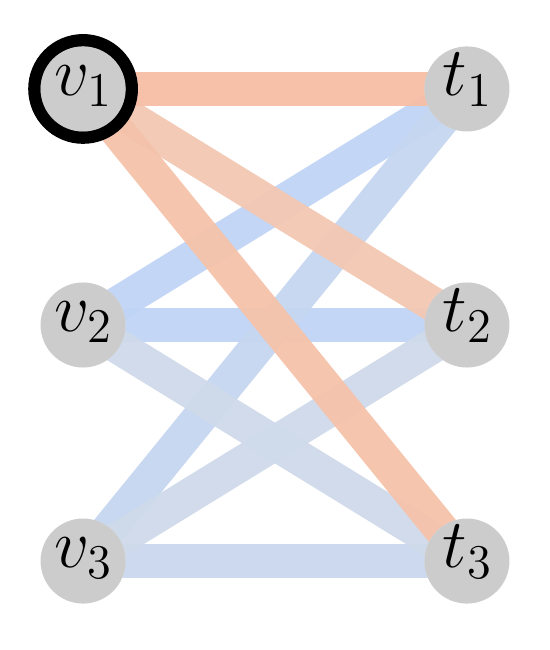}
		\subcaption{Result of $v_1$}
		\label{fig:grad1}
	\end{minipage}
	\begin{minipage}[t]{.2\textwidth}
		\includegraphics[width=1.0\textwidth]{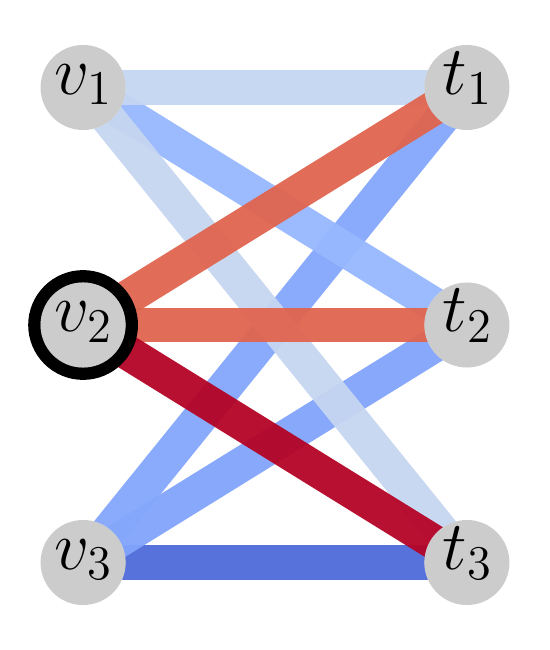}
		\subcaption{Result of $v_2$}
		\label{fig:grad2}
	\end{minipage}
	\begin{minipage}[t]{.2\textwidth}
		\includegraphics[width=1.0\textwidth]{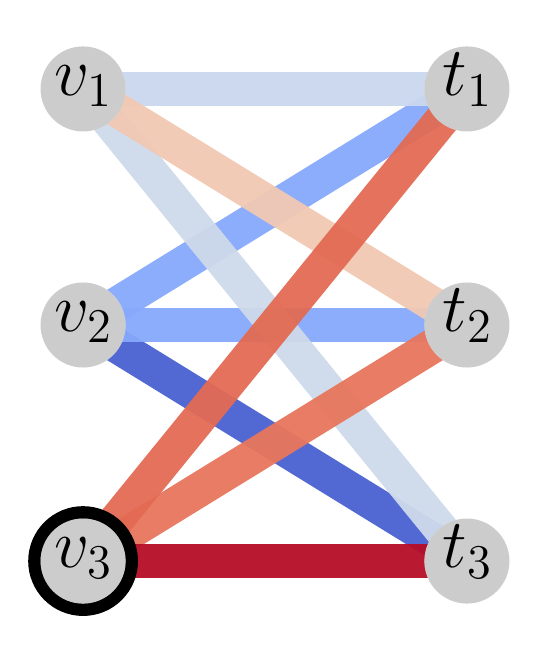}
		\subcaption{Result of $v_3$}
		\label{fig:grad3}
	\end{minipage}
	\begin{minipage}[t]{.08\textwidth}
		\includegraphics[width=1.0\textwidth]{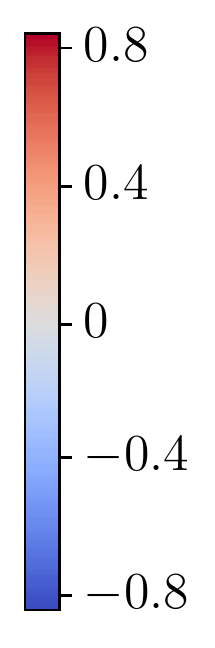}
		\subcaption*{}
		\label{fig:cbar}
	\end{minipage}
	\caption{(a): 
		Given $\thetavec$ values.  
		Thick and thin edges have link probabilities 
		$0.4$ and $0.2$, respectively. 
		(b) to (d): 
		Sensitivity analysis results.  
		Edge colors in (b), (c), and (d) indicate how the increase in the corresponding $\thetavec$ entries can affect the probability of choosing $v_1$, $v_2$, and $v_3$, respectively.  
	}
\end{figure}

\subsection{Sensitivity analysis}\label{subsec:exp_sensitivity}
We perform sensitivity analysis with a synthetic instance such that $V = \{v_1, v_2, v_3\}$,  $T = \{t_1, t_2, t_3\}$, 
and $K=2$. 
Let $\theta_{i,j}$ denote the link probability of $(v_i, t_j)$; 
we set 
$(\theta_{1,1}, \theta_{1,2}, \theta_{1,3}) = (0.4, 0.4, 0)$, 
$(\theta_{2,1}, \theta_{2,2}, \theta_{2,3}) = (0, 0.4, 0.2)$, and 
$(\theta_{3,1}, \theta_{3,2}, \theta_{3,3}) = (0, 0, 0.2)$ as in \Cref{fig:prob}. 
We analyze the sensitivity of {\sgreedy} 
by estimating $\nabla_\thetavec \E_{S\sim p(\thetavec)}[\ib_S]$ as explained in \Cref{subsec:sensitivity}.  
We let $N = 100$ and 
reduce the variance with 
the baseline correction method explained in \Cref{sec:gradient}.

Figures \ref{fig:grad1}, \ref{fig:grad2}, and \ref{fig:grad3} illustrate 
how and how much the increase in each $\theta_{i,j}$ value can affect the probability of choosing $ v_1$, $v_2$, and $ v_3$, respectively. 
In this setting, 
the objective values of 
the three maximal solutions, $\{v_1, v_2\}$, $\{v_1, v_3\}$, and $\{v_2, v_3\}$, 
are $1.24$, $1.00$, and $0.76$, respectively. 
Therefore, {\sgreedy} returns $\{v_1, v_2\}$ or $\{v_1, v_3\}$ with a high probability; 
this remains true even if the $\thetavec$ values slightly change. 
Thus, the probability of choosing $v_1$ is relatively insensitive 
as in \Cref{fig:grad1}. 
In contrast, 
as in \Cref{fig:grad2,fig:grad3}, 
the probabilities of choosing $v_2$ and $v_3$, 
respectively, are highly sensitive. 
For example, if $\theta_{2,3}$ increases, 
the probability that the algorithm returns $\{v_1, v_2\}$ ($\{v_1, v_3\}$) increases (decreases), 
which means the probability of choosing $v_2$ ($v_3$) is positively (negatively) affected by the increase in $\theta_{2,3}$. 
We can also see the that the opposite occurs if $\theta_{3,3}$ increases.

%        In [26]: set_func([0,1], P, w)                                                                                                                                     
%        Out[26]: 1.2400000095367432
%        
%        In [27]: set_func([0,2], P, w)                                                                                                                                     
%        Out[27]: 0.9999999403953552
%        
%        In [28]: set_func([1,2], P, w)                                                                                                                                     
%        Out[28]: 0.7599999308586121

\subsection{Decision-focused learning}\label{subsec:exp_decision}

\begin{table}
	\centering
	\caption{Function values achieved with each method.}
	\label{table:experiment_decision}
	{\small
		\begin{tabular}{rcccccc}
			\toprule
			& \multicolumn{2}{c}{$K=5$} & \multicolumn{2}{c}{$K=10$} & \multicolumn{2}{c}{$K=20$} \\
			\cmidrule(r){2-3} \cmidrule(r){4-5} \cmidrule(r){6-7} 
			& Training & Test & Training & Test & Training & Test \\
			\midrule
			SG-$1$      & $26.3 \pm 4.0$ & $26.4 \pm 4.4$ & $46.0 \pm  5.9$ & $45.9 \pm  6.5$ & $69.7 \pm 23.8$ & $69.6 \pm 24.1$ \\
			SG-$10$     & $29.0 \pm 3.7$ & $28.1 \pm 4.9$ & $47.0 \pm 12.1$ & $46.1 \pm 12.4$ & $71.5 \pm 28.0$ & $70.6 \pm 28.1$ \\
			SG-$100$    & $33.6 \pm 2.4$ & $32.0 \pm 3.8$ & $54.3 \pm  2.0$ & $53.5 \pm  4.2$ & $82.6 \pm 21.8$ & $82.3 \pm 21.7$ \\
%			VR-SG-$1$   & $28.4 \pm 0.8$ & $28.0 \pm 2.4$ & $50.3 \pm  1.5$ & $50.4 \pm  2.8$ & $86.7 \pm  1.4$ & $86.3 \pm  3.6$ \\
			VR-SG-$10$  & $35.2 \pm 6.1$ & $33.7 \pm 6.2$ & $57.9 \pm  1.6$ & $56.2 \pm  3.4$ & $90.8 \pm 16.5$ & $89.5 \pm 16.7$ \\
			VR-SG-$100$ & $\bf 36.8 \pm 0.9$ & $\bf 35.6 \pm 2.2$ & $\bf 59.9 \pm  1.6$ & $\bf 58.0 \pm  2.9$ & $\bf 96.8 \pm  1.1$ & $\bf 94.5 \pm  2.6$ \\
			Continuous         & $24.0 \pm 4.5$ & $23.2 \pm 4.9$ & $43.2 \pm  6.1$ & $42.3 \pm  7.1$ & $81.7 \pm  6.8$ & $81.3 \pm  6.6$ \\
			Two-stage   & $17.3 \pm 1.2$ & $17.3 \pm 2.1$ & $35.6 \pm  0.9$ & $35.6 \pm  2.7$ & $65.5 \pm  4.0$ & $64.8 \pm  5.1$ \\
			Random      & $17.5 \pm 1.0$ & $17.6 \pm 2.2$ & $33.8 \pm  0.8$ & $34.0 \pm  2.7$ & $64.0 \pm  1.3$ & $64.5 \pm  2.6$ \\
			\bottomrule
		\end{tabular}
	}
\end{table}

We evaluate the performance of our method via decision-focused learning experiments 
with MovieLens $100$K dataset \citep{harper2015movielens}, which contains $100,000$ ratings ($1$ to $5$) of $1,682$ movies made by  $943$ users. 
We set the link probabilities at $0.02, 0.04, \dots, 0.1$ according to the ratings; those of unrated ones are set at $0$. 
We randomly sample $100$ movies and $500$ users, which form item set $V$ and target set $T$, respectively.  
We thus make $100$ random $(V, T)$ pairs with link probabilities.  
Each movie $v\in V$ belongs to some of $19$ genres, e.g., action and horror; 
we use the $19$-dimensional indicator vector as a movie feature. 
Each user $t\in T$ has information of their age, sex, and occupation categorized into $21$ types, 
e.g., writer and doctor; 
we concatenate them and use the resulting $24$-dimensional vector as a user feature. 
A feature of each $(v,t)\in V \times T$ is a concatenation of the $19$- and $24$-dimensional vectors.  
As a result, each of the $100$ random $(V,T)$ pairs has feature $\Xb$ of form $100\times 500 \times 43$. 
The predictive model,
which outputs $\theta_{v,t}\in[0,1]$ for the feature of each $(v,t)\in V\times T$, 
is a $2$-layer NN with a hidden layer of size $200$ and ReLU activation functions; 
the outputs are clipped to $[0,1]$. 
Since the features are sparse, the predictive model with default weight initialization 
returns $0$ too frequently; 
to avoid this, 
we set initial linear-layer weights at random non-negative values drawn from $[0, 0.01]$. 

We split the $100$ random instances into $80$ training and $20$ test instances; 
we train the predictive model with $(\Xb_1, \thetavec_1), \dots, (\Xb_{80}, \thetavec_{80})$ 
and test the performance with $(\hat \Xb_1, \hat \thetavec_1), \dots, (\hat \Xb_{20}, \hat \thetavec_{20})$. 
We make $30$ random training/test splits, and  
we present all results with means and standard deviations over the $30$ random splits. 
Given $80$ training datasets, 
we train the model over mini-batches of size $20$ for $5$ epochs.
We use Adam with learning rate $10^{-3}$ for updating the model parameter, $\wvec$.\footnote{
	The settings mostly replicate those of budget allocation instances in \citep{wilder2019melding}, 
	but we use the public MovieLens dataset instead of the original one, which is not open to the public. Accordingly, some parts are slightly changed.}  

We compare 
{\bf SG-$N$}, 
{\bf VR-SG-$N$}, 
{\bf Continuous}, 
{\bf Two-stage}, and 
{\bf Random}. 
{\bf SG-$N$} is our method based on {\sgreedy} (see, \Cref{subsec:decision}), 
where $N$ indicates the number of output samples; we let $N=1$, $10$, and $100$. 
{\bf VR-SG-$N$} 
(variance-reduced {\bf SG-$N$}) uses the baseline correction method when estimating gradients; we let $N=10$ and $100$ (omit $N=1$) since if $N=1$, the baseline value is equal to the single output value, which always yields zero gradients. 
Both {\bf SG-$N$} and {\bf VR-SG-$N$} use the greedy algorithm when making decisions. 
{\bf Continuous} \citep{wilder2019melding} 
maximizes the continuous relaxation of the objective function with SGA 
and differentiates local optima;  
we use their original implementation. 
{\bf Two-stage} trains the model by minimizing the mean square error, and then maximizes the objective function with SGA;  
the implementation is based on that of \citep{wilder2019melding}.  
{\bf Continuous} and {\bf Two-stage} make decisions $S\in \Ical$ by choosing elements corresponding to 
the top-$K$ entries of solution $\xvec\in[0,1]^{n}$ returned by SGA. 
{\bf Random} is a baseline method that makes uniformly random decisions $S \in\Ical$.

\Cref{table:experiment_decision} shows the 
objective function values 
(averaged over the $80$ training and $20$ test instances) 
achieved by each method 
for $K=5$, $10$, and $20$. 
{\bf VR-SG-$100$} achieves the highest objective value for every case, 
and 
({\bf VR-}){\bf SG} with other settings also performs comparably to or better than {\bf Continuous}. 
These results are consistent with the theoretical guarantees. 
More precisely, 
while {\bf Continuous} trains the predictive model so that SGA, a $1/2$-approximation algorithm, returns high objective values, 
our methods train the model so that the (almost) ($1-1/\e$)-approximation 
(smoothed) greedy algorithm can achieve high objective values. 
We can also see that the variance reduction method is effective for improving the performance of our method. 
The standard deviation of ({\bf VR-}){\bf SG} becomes sometimes high; 
this is because they are sometimes trapped in poor local optima and result in highly deviated objective values. 
Considering this, the performance of our method would be further improved if we can combine it with NN training techniques for escaping from poor local optima. 
Regarding running times, 
for updating $\wvec$ once, {\bf SG}-$1$ takes $2.81$, $3.38$, and $3.77$ seconds on average for $K=5$, $10$, and $20$, respectively,  
while {\bf Continuous} takes $5.86$, $5.87$, and $6.11$ seconds, respectively.  
Hence, our methods can run faster by performing {\sgreedy} in parallel as mentioned in \Cref{subsec:decision}.

%time 
%K = 5   grd=2.8112858057022097 cnt=5.861440670490265
%K = 10 grd=3.385388708114624 cnt=5.87943457365036
%K = 20 grd=3.778781473636627 cnt=6.112937152385712

\section*{Broader Impact}
The greedy algorithm for submodular maximization is one of the most extensively studied subjects of combinatorial optimization in the machine learning (ML) community. On the other hand, many recent advances in ML methods are based on continuous optimization; particularly, NNs are usually trained with stochastic first-order methods. Our work, which serves as a bridge between the combinatorial greedy algorithm and continuous first-order methods, will benefit researchers in the optimization and ML communities and practitioners who have ML tasks related to submodular maximization. 
Below we present examples of practical situations where our framework is useful. 
 
\begin{itemize}
	\item 
	Submodular maximization sometimes appears when making vital decisions; e.g., allocation of large resources to advertising channels. In such situations, after computing a solution with the greedy algorithm, we can use our sensitivity analysis method (see, \Cref{subsec:sensitivity,subsec:exp_sensitivity}) for assessing the reliability of the solution, i.e., how robust it is against uncertainties in parameters of objective functions. If the entries of the estimated Jacobian corresponding to the solution are small enough in absolute value, then the solution is reliable and thus we can put it into practice; otherwise, we can try using robust submodular maximization methods (e.g., \citep{staib2019distributionally}) to strike a balance between the objective value and reliability. 

	\item
	ML tasks related to submodular maximization often involve prediction of parameters; for example, when designing diverse-recommendation systems, we need to predict users' preferences, which correspond to the parameters. In such situations, the decision-focused learning method based on our framework (\Cref{subsec:decision,subsec:exp_decision}) is useful, particularly when we do not have enough prior knowledge on how to design good predictive models and loss functions. Note that the simplicity of our method, which does not use the multilinear extension unlike \citep{wilder2019end}, is also beneficial to practitioners. 	  
	
\end{itemize}

As a negative aspect, failures of systems that utilize our method may result in harmful consequences. In particular, when our method is combined with NNs, how to avoid poor local optima is a practically important issue as mentioned in \Cref{subsec:exp_decision}. 
To resolve this, we need to study the structures of objective functions that appear in each situation in detail, which we leave for future work.

\bibliographystyle{plainnat}
\bibliography{mybib}

\clearpage

\appendix
\setcounter{equation}{0}
\renewcommand{\theequation}{A\arabic{equation}}

\clearpage

%\onecolumn
\begin{center}
	{\fontsize{18pt}{0pt}\selectfont \bf Appendix}
\end{center}

\section{Comparisons with existing greedy methods}\label{sec:comparison}
We present detailed comparisons of 
our work and 
the exiting studies \citep{tschiatschek2018differentiable,powers2018differentiable} on the differentiable greedy methods, 
which use $\softmax$ instead of $\argmax$.  
As explained below, 
the existing methods are devoted to differentiating some functions defined with subsets $X_1, X_2,\ldots\subseteq V$ given as training data. 
In contrast, 
we do not assume such subsets to be given and consider differentiating the expected value of any output-dependent quantities, $\E_{S\sim p(\thetavec)}[Q(S)]$; 
note that this design of our framework is the key to dealing with sensitivity analysis and decision-focused learning.  
Our framework can also provide more reasonable approaches to their problem settings as described below.

\citet{tschiatschek2018differentiable} consider differentiating the likelihood function, 
which quantifies how close an output of their algorithm can be to some good solutions, $X_1, X_2,\dots$, given as training data. 
To this end, we need to differentiate $P(X) \coloneqq \Sigma_{\sigma \in \Sigma(X)} P(\sigma, \thetavec)$, 
where $X\in \{X_1,X_2,\dots\}$ is a given subset, $\Sigma(X)$ is the set of all permutations of elements in $X$, and $P(\sigma, \thetavec)$ is the probability that their algorithm returns sequence $\sigma\in\Sscr_{\le K}$.  
Since the computation of the summation over $\Sigma(X)$ is too costly, they employ the following heuristic approximation: 
if the temperature of $\softmax$ is low, 
we let $P(X)\approx P(\sigma^G,\thetavec)$, where $\sigma^G$ is obtained by the greedy algorithm, 
and if the temperature is high, 
we let $P(X) \approx |X|! \times P(\sigma^R, \thetavec)$, 
where $\sigma^R$ is a random permutation. 
As a result, the computed derivative has no theoretical guarantees 
unlike our gradient estimator, which is guaranteed to be unbiased.   
Note that with our method, 
we can compute an unbiased estimator of the desired derivative as follows: 
we let $Q(S)$ return $1$ if $S$ and $X$ consist of the same elements and $0$ otherwise, and we estimate $
\nabla_\thetavec P(X) = \nabla_\thetavec \E_{S \sim p(\thetavec)} [Q(S)]$ 
as explained in \Cref{sec:gradient}.   

\citet{powers2018differentiable} focus on some cases where 
we can compute derivatives more easily. 
They consider some loss function $L(X, \pvec_1(\thetavec), \dots, \pvec_K(\thetavec))$ 
that is differentiable w.r.t. $\pvec_1(\thetavec), \dots, \pvec_K(\thetavec)$, 
where $X$ is given as training data. 
In their setting, $\pvec_i(\thetavec)$ is given by $\softmax$ and 
$\f(X, \thetavec)$ is differentiable w.r.t. $\thetavec$. 
Thus,  once $X$ is fixed, 
$\nabla_\thetavec L(X, \pvec_1(\thetavec), \dots, \pvec_K(\thetavec))$ 
can be readily computed via automatic differentiation.  
%This setting is completely different from ours; 
%for example, they do not take the randomness of outputs, which are distributed over $\Sscr_{\le K}$, into account.  
%
From the perspective of our method, we can regard their method as the one dealing with the case of  $N=1$. 
More precisely, if we take $X$ to be a single output of {\sgreedy} 
and 
let 
$L(X, \pvec_1(\thetavec), \dots, \pvec_K(\thetavec)) = Q(X)\ln p(X,\thetavec)$, 
then 
the derivative computed by their method coincides with the one obtained by using our method with $N=1$. 
Note that 
the above design of $L(\cdot)$, which is the key to obtaining unbiased gradient estimators, 
and the case of $N>1$ are not studied in \citep{powers2018differentiable}.  
Furthermore, if we apply our framework to their problem setting, we can use (non-differentiable) loss functions,  $L(X, S)$, 
that measure the distance between given $X$ and output $S$ (e.g., Hamming and Levenshtein distances); 
we let $Q(S)=L(X, S)$ and estimate $\nabla_\thetavec\E_{S\sim p(\thetavec)}[L(X, S)]$.

\section{Proofs of approximation guarantees}\label{asec:proofs}
In the following discussion, $S_k$ denotes the solution constructed in the $k$-th step of  
\Cref{alg:pgreedy}; 
we let $S_0=\emptyset$. 
For simplicity, we omit the fixed parameter, $\thetavec$, in the proofs. 
%The following discussion holds as long as $\f(\cdot, \thetavec)$ is normalized, monotone, and submodular. 

\marginal*

\begin{proof}
	From the rule of choosing $s_k$, we have $\E[\fdel{s_k}{S}] = \iprod{\gvec_k , \pvec_k}$. 
	Let $\ib_u\in\R^\nk$ be the indicator vector of $u\in \Uk$. 
	Since $\iprod{\gvec_k, \ib_u} = \fdel{u}{S}$ and $\ib_u\in \Delta^\nk$ hold, 
	we can obtain the lemma as follows:  
	\begin{align}
	\E[\fdel{s_k}{S}] = \iprod{\gvec_k , \pvec_k}    
	&= 
	\max_{\pvec\in\Delta^\nk} \{\iprod{\gvec_k, \pvec} - \Omega_k(\pvec)\} + \Omega_k(\pvec_k) \\
	&\ge
	\iprod{\gvec_k, \ib_u}  - (\Omega_k(\ib_u) - \Omega_k(\pvec_k))
	\ge
	\fdel{u}{S} - \delta, 
	\end{align}
	where the last inequality comes from 
	$\delta\ge \Omega_k(\pvec) - \Omega_k(\qvec)$ 
	for any $\pvec, \qvec\in \Delta^{\nk}$.
\end{proof}

\cardinality*

\begin{proof}
	Fix $k \in \{ 1,\dots, K \}$ arbitrarily 
	and take all random quantities to be 
	conditioned on the ($k-1$)-th step.
	From \Cref{lem:fdel} with $O\backslash\Sk{k-1} \subseteq \Uk$ 
	and the submodularity, we obtain
	\[
	\E[\fdel{s_k}{\Sk{k-1}}] \ge \frac{1}{K} \sum_{v\in O\backslash S_{k-1}} \fdel{v}{\Sk{k-1}} -\delta 
	\ge \frac{1}{K} \fdel{O}{\Sk{k-1}} - \delta.   
	\] 
	By taking expectation over all possible realizations of the ($k-1$)-th step and 
	using the monotonicity, we obtain
	\[
	\E[\f(\Sk{k})] - \E[\f(\Sk{k-1})]    
	\ge \frac{1}{K} (\E[\f(O \cup \Sk{k-1})] - \E[\f(\Sk{k-1})]) - \delta
	\ge \frac{1}{K} (\f(O) - \E[\f(\Sk{k-1})]) - \delta.
	\] 
	Therefore, as is often the case with the analysis of the greedy algorithm, 
	we can obtain the following inequality by induction:   
	\[ 
	\E[\f(\Sk{K})] \ge \left( 1 - \left(1 - \frac{1}{K}\right)^K \right) \f(O) 
	- \delta\sum_{k=0}^{K-1}\left(1 - \frac{1}{K}\right)^k
	\ge \left(1 - \frac{1}{\e}\right) \f(O) - \delta K,  
	\]
	where we used $\f(\emptyset)=0$. 
	Hence we obtain the theorem from $\E[\f(S)] = \E[\f(S_K)]$. 
\end{proof}

\extensible*

\newcommand{\Kp}{{|S|}}
\begin{proof}
	For each realization of $S_0 \subset S_1 \subset \dots \subset S_\Kp  = S \in \Ical$, 
	we define $S_{\Kp +1}, S_{\Kp +2} , \dots, S_{K}$ as $S$. 
	We thus construct a series of feasible solutions, $S_0, S_1, \dots, S_K$, 
	for every realization.  
	Note that $\E[\f(S)] = \E[\f(S_K)]$ holds since we always have $S = S_K$.

	We consider constructing a series of subsets $O_0,O_1,\dots,O_K$ 
	for each realization of $S_0, S_1, \dots, S_K$. 
	We aim to prove that we can construct such $O_0, O_1, \dots,O_K$ satisfying the following conditions: 
	$O_0 = O$, 
	$S_i \subseteq O_i \in\Ical$ ($i=0,\dots, K-1$), 
	$S_K = O_K \in \Ical$ for every realization, 
	and 
	\begin{align}\label{eq:telescope}
	\kappa \cdot (\E[\f(S_{i})] - \E[\f(S_{i-1})] + \delta) \ge \E[\f(O_{i-1})] - \E[\f(O_i)]
	\end{align}    
	for $i=1,\dots,K$. 
	
	In the case of $i=0$, we let $O_0 = O$, 
	which satisfies $S_0 = \emptyset \subseteq O = O_0 \in\Ical$. 
	In this case, \eqref{eq:telescope} is not required to hold. 
	
	We assume all random quantities to be conditioned on an arbitrary realization of the ($k-1$)-th step, 
	where $S_0,\dots, S_{k-1}$ and $O_0,\dots,O_{k-1}$ satisfying 
	$S_i\subseteq O_i \in \Ical$ ($i=0,\dots, k-1$) are given. 
	If $S_{k-1}$ is maximal, we let $O_k = S_k$ ($= S_{k-1} = O_{k-1}$), 
	which satisfies $S_k = O_k \in \Ical$ and 
	\[
	\kappa \cdot (\E[\f(S_k)] - \f(S_{k-1}) + \delta) = \kappa \cdot \delta \ge 0 = \f(O_{k-1}) - \E[\f(O_k)]. 
	\]
	If $S_{k-1}$ is not maximal, from the definition of $\kappa$-extensible systems, 
	for any choice of $s_k\notin S_{k-1}$, there exists $Z_k\subseteq O_{k-1}\backslash S_{k-1}$ such that 
	$O_{k-1}\backslash Z_k \cup \{ s_k \}  \in\Ical$ and $|Z_k| \le \kappa$ hold.  
	We let $O_{k} = O_{k-1}\backslash Z_k \cup \{ s_k \}$. 
	Note that thus constructed $O_k$ satisfies
	$S_k \subseteq O_k \in \Ical$ for any realization of the $k$-th step; 
	moreover, if $k =K$, we always have $S_K = O_K \in\Ical$ since $S_K$ is maximal in any realization. 
	Considering expectation over realizations of the $k$-th step, we obtain 
	\begin{align}
	&\f(O_{k-1}) - \E[\f(O_k)] \\
	={}&
	\f(O_{k-1}) - \E[\f(O_{k-1} \backslash Z_k)] \\
	&\qquad \quad \ \ \,  + \E[\f(O_{k-1} \backslash Z_k)] - \E[\f(O_k)] \\    
	\le{}&
	\E[\fdel{Z_k}{O_{k-1} \backslash Z_k}] 
	&\because\text{$O_{k-1} \backslash Z_k \subseteq O_k$ and monotonicity} \\
	\le{}&
	\E\left[\sum_{v\in Z_k} \fdel{v}{S_{k-1}}\right]
	&\because\text{$S_{k-1}\subseteq O_{k-1}\backslash Z_k$ and submodularity} \\
	\le{}& 
	\kappa \cdot (\E[\fdel{s_k}{S_{k-1}}] + \delta)   
	&\because\text{$Z_k \subseteq O_{k-1}\backslash S_{k-1} \subseteq \Uk$, \Cref{lem:fdel}, and $|Z_k|\le \kappa$} \\
	={}&
	\kappa \cdot (\E[\f(S_k)] - \f(S_{k-1}) + \delta)
	\end{align}
	Therefore, in any case we have 
	\[
	\kappa \cdot (\E[\f(S_k)] - \f(S_{k-1}) + \delta) \ge \f(O_{k-1}) - \E[\f(O_k)].  
	\]
	By taking expectation over all realizations of the ($k-1$)-th step, we obtain \eqref{eq:telescope} for $i=k$. 
	For every realization, 
	$O_0,\dots, O_k$ constructed above satisfy 
	$S_i \subseteq O_i \in \Ical $ for $i = 0,\dots, k$ 
	(if $k = K$, we have $S_k = O_k \in \Ical$). 
	This means that the assumption of induction for the next step is satisfied.  	
	Consequently, \eqref{eq:telescope} holds for $i=1,\dots, K$ by induction.  
	Summing both sides of \eqref{eq:telescope} for $i=1,\dots, K$, 
	we obtain 
	\[
	\kappa \cdot (\E[\f(S_K)] - \f(\emptyset) + \delta K ) \ge  \E[\f(O_0)] - \E[\f(O_K)].
	\]
	Since we have $\f(\emptyset)=0$, $O_0 = O$, and $O_K = S_K$ for every realization, it holds that
	\[
	\E[\f(S_K)] \ge \frac{1}{\kappa+1} \f(O) - \frac{\kappa}{\kappa+1} \delta K 
	\ge \frac{1}{\kappa+1} \f(O) -\delta K.
	\]    
	Hence we obtain the theorem from $\E[\f(S)] = \E[\f(S_K)]$. 
\end{proof}

\section{Regularization functions}\label{sec:regularization}
We first detail the case where $\Omega_k$ is the entropy function. 
We then present a sufficient condition for satisfying \Cref{assump:diff_p}, 
which is useful when designing regularization functions. 

\subsection{Entropy regularization}\label{subsec:entropy}
We consider using the entropy function as a regularization function: 
$
\Omega_k (\pvec) = \epsilon \sum_{u \in \Uk} p(u) \ln  p(u)
$, 
where $\epsilon>0$ is a constant that controls the perturbation strength. 
Note that we have
$\Omega_k(\pvec) - \Omega_k(\qvec) \le \epsilon\cdot0 - \epsilon\sum_{i=1}^\nk \frac{1}{\nk}\ln\frac{1}{\nk} = \epsilon\ln\nk$ 
for any $\pvec, \qvec\in \Delta^\nk$.

From the relationship between the entropy regularization and $\softmax$, 
each iteration of {\sgreedy} can be performed via softmax sampling. 
More precisely, from the Karush--Kuhn--Tucker (KKT) condition of problem \eqref{eq:solveomega}, 
$\max_{\pvec\in\Delta^{\nk}}\{\iprod{\gveck{k}, \pvec} - \Omega_k(\pvec)   \}$, we have 
\begin{align}\label{eq:kkt}
\epsilon
(\ln \pvec + \ib_\nk)
- \gvec_k + \ib_\nk \mu = \zeros_\nk 
\quad \text{and} \quad 
\ib_\nk^\top \pvec  = 1,  
\end{align}
where $\ln$ operates in an element-wise manner 
and $\mu\in\R$ is a multiplier corresponding to the equality constraint. 
Note that we need not take the inequality constraints, $\pvec\ge\zeros_\nk$, into account since the entropy regularization 
forces every $p(u)$ to be positive. 
Since $\Omega_k$ is strictly convex and every feasible solution satisfies the linear independence constraint qualification (LICQ), 
the maximizer, $\pvec_k$, is characterized as the unique solution to the KKT equation system \eqref{eq:kkt}.  
From \eqref{eq:kkt}, we see that $\pvec_k$ is proportional to $\exp(\gvec_k/ \epsilon)$.  
Thus, Steps \ref{step:g} to \ref{step:choose} in \Cref{alg:pgreedy} can be 
performed via softmax sampling: $p_k(u, \thetavec) \propto \exp(\fdelp{u}{\Sk{k-1}}{\thetavec} / \epsilon)$ for $u\in \Uk$, 
which takes $\mathrm{O}(\nk)$ time if $\fdelp{u}{\Sk{k-1}}{\thetavec}$ values are given.

We then discuss how to compute $\nabla_{\gvec_k} \pveck{k}(\gvec_k)$.  
While this can be done by directly differentiating $p_k(u, \gvec_k) \propto \exp(g_k(u) / \epsilon)$, 
we here see how to compute it by applying the implicit function theorem 
(see, e.g., \citep{dontchev2014implicit}) 
to the KKT equation system \eqref{eq:kkt} 
as a warm-up for the next section.
In this case, the requirements for using the implicit function theorem 
are satisfied (see the next section). 
By differentiating the KKT equation system \eqref{eq:kkt} w.r.t. $\gvec_k$, we obtain 
\begin{align}
\begin{bmatrix}
\epsilon \diag(\pvec_k)^{-1} & \ib_\nk \\
\ib_\nk^\top & 0
\end{bmatrix} 
\begin{bmatrix}
\nabla_{\gvec_k} \pvec_k \\
\nabla_{\gvec_k} \mu
\end{bmatrix} 
= 
\begin{bmatrix}
\Ib_\nk \\ 
\zeros_\nk^\top
\end{bmatrix},
\end{align}
where $\diag(\pvec_k)$ is a diagonal matrix whose diagonal entries are $\pvec_k$ 
and $\Ib_\nk$ is the $\nk\times \nk$ identity matrix. 
We can compute $\nabla_{\gvec_k} \pveck{k}$ by solving the above equation as follows: 
\begin{align}
\begin{bmatrix}
\nabla_{\gvec_k} \pvec_k \\
\nabla_{\gvec_k} \mu
\end{bmatrix} 
= 
\begin{bmatrix}
\epsilon \diag(\pvec_k)^{-1} & \ib_\nk \\
\ib_\nk^\top & 0
\end{bmatrix}^{-1} 
\begin{bmatrix}
\Ib_\nk \\ 
\zeros_\nk^\top
\end{bmatrix}
=
\begin{bmatrix}
\epsilon^{-1} (\diag(\pvec_k) - \pvec_k \pvec_k^\top) \\
\pvec_k^\top 
\end{bmatrix}. 
\end{align}
Note that once we obtain
%$\gvec_k(\thetavec)$ and 
$\nabla_\thetavec \gvec_k(\thetavec)$, 
we can compute the desired derivative, 
$
\nabla_\thetavec \pveck{k}(\thetavec)
=
\nabla_{\gvec_k} \pveck{k}(\gvec_k) \cdot \nabla_\thetavec \gvec_k(\thetavec)
=
\epsilon^{-1} (\diag(\pvec_k) - \pvec_k \pvec_k^\top) \nabla_\thetavec \gvec_k(\thetavec)
$, 
by matrix-vector products in $\mathrm{O}(\nk \times \dim\Theta)$ time. 

One may get interested in the link between 
problem \eqref{eq:solveomega} with the entropy regularization and 
the optimal transport (OT) with entropy regularization \citep{cuturi2013sinkhorn}. 
Specifically, while \eqref{eq:solveomega} has a vector variable with one equality constraint, OT has a matrix variable with two equality constraints; in this sense, \eqref{eq:solveomega} considers a simpler setting. 
Thanks to the simplicity, we can analyze the theoretical guarantees of {\sgreedy}.  
In contrast, if we consider using OT, we can employ more sophisticated operations, 
e.g., ranking and sorting \citep{cuturi2019differentiable}, 
than $\argmax$. 
In return for this, however, it becomes more difficult to prove approximation guarantees; for example, how to design transportation costs is non-trivial. 
This OT-based approach to designing differentiable combinatorial optimization algorithms will be an interesting research direction, which we leave for future work.

\subsection{Sufficient condition for satisfying \Cref{assump:diff_p}}\label{subsec:perturb_assump}\label{subsec:sufficient}
We study the case where $\Omega_k$ is a general strictly convex differentiable function; 
although a similar discussion is presented in \citep{amos2017optnet} 
for the case where $\Omega_k$ is quadratic, we here provide a detailed analysis with general $\Omega_k$ 
for completeness. 
The KKT condition 
of problem \eqref{eq:solveomega} can be written as 
\begin{align}
\nabla_\pvec \Omega_k(\pvec)  - \gvec_k - \lambdavec + \ib_\nk \mu 
= \zeros_\nk, & & 
\lambdavec \odot \pvec 
= \zeros_\nk, & & 
\text{and} & &
\ib_\nk^\top \pvec 
= 1,
\end{align}
where 
$\lambdavec\ge \zeros_\nk$ 
consists of multipliers corresponding to the inequality constraints, 
$\pvec\ge\zeros_\nk$, 
and 
$\odot$ denotes the element-wise product. 
Since every feasible point in $\Delta^\nk$ satisfies LICQ, 
if
$\Omega_k$ is strictly convex on $\Delta^\nk$, 
the optimal solution is uniquely characterized by 
the KKT condition. 
Let $(\tilde\pvec, \tilde\lambdavec, \tilde\mu)$ be 
a triplet that satisfies the KKT condition, 
where $\tilde{\pvec} = \pvec_k$.  
If the following three conditions hold, 
$\nabla_{\gvec_k} \pvec_k$ can be calculated from the KKT condition 
as detailed later: 
%%
%$\Omega_k$ is twice-differentiable, 
%the Hessian, 
%$\nabla_\pvec^2 \Omega_k(\tilde\pvec)$, 
%is positive definite, 
%and 
%the strict complementarity, $\tilde\lambdavec + \tilde\pvec > \zeros_\nk$, 
%holds, 
%%
%Summarizing the above, 
%a sufficient condition for 
%satisfying \Cref{assump:diff_p} can be written as follows: 
\begin{enumerate}
	\item $\Omega_k$ is twice-differentiable, 
	\item the Hessian, $\nabla_\pvec^2\Omega_k(\pvec)$, is positive definite for any $\pvec\in\Delta^\nk$, and 
	\item the strict complementarity, 
	$\tilde\lambdavec + \tilde\pvec > \zeros_\nk$, 
	holds at the unique optimum, $\tilde\pvec = \pvec_k$.      
\end{enumerate}
Note that the second condition implies the strict convexity of $\Omega_k$ 
on $\Delta^{\nk}$. 
Therefore, a sufficient condition for satisfying \Cref{assump:diff_p} 
is given by the above three conditions. 
In practice, given any twice-differentiable convex function, 
we can add to it the entropy function multiplied by a small constant 
for obtaining $\Omega_k$ that satisfies the sufficient condition.

\newcommand{\U}{{\tilde U}}
We then explain how to compute $\nabla_{\gvec_k}\pvec_k$. 
Let $\U$ be a subset of $\Uk$ such that $\tilde{p}(u) = 0$ iff $u \in\U$; 
the strict complementarity implies $\tilde{\lambda}(u) = 0$ iff $u \notin \U$.  
We let $\xvec \coloneqq (\pvec, \lambdavec_\U, \mu)$, 
where 
$\lambdavec_\U$ is a $|\U|$-dimensional vector consisting of 
the entries of $\lambdavec$ corresponding to $\U$.
We define $\Ib_\U$ as the $\nk\times|\U|$ matrix that has columns of $\Ib_\nk$ corresponding to $\U$. 
The KKT equation system 
at $\tilde{\xvec} = (\tilde\pvec, \tilde\lambdavec_\U, \tilde\mu)$ can be 
written as 
\begin{align}
H(\xvec, \gvec_k) \coloneqq 
\begin{bmatrix}
\nabla_\pvec \Omega_k(\pvec)  - \gvec_k - \Ib_\U  \lambdavec_\U + \ib_\nk \mu 
\\
-\Ib_\U^\top \pvec 
\\
\ib_\nk^\top \pvec - 1
\end{bmatrix}
=
\begin{bmatrix}
\zeros_\nk \\ \zeros_{|\U|} \\ 0
\end{bmatrix},
\end{align}
and its partial Jacobians at $\tilde{\xvec}$ are given by
\begin{align}
\nabla_\xvec H(\tilde{\xvec}, \gvec_k) =
\begin{bmatrix}
\nabla_\pvec^2 \Omega_k(\tilde{\pvec}) & - \Ib_\U & \ib_\nk 
\\
- \Ib_\U^\top & \multicolumn{2}{c}{\multirow{2}{*}{ $\zeros_{|\U| + 1 \times |\U| + 1}$ }}
\\
\ib_\nk^\top &
\end{bmatrix}
& & \text{and} 
& & 
\nabla_{\gvec_k} H(\tilde{\xvec}, \gvec_k) 
= 
\begin{bmatrix}
-\Ib_\nk \\ \zeros_{|\U| \times \nk} \\ \zeros_\nk^\top
\end{bmatrix}.
\end{align}
Note that $|\U|<\nk$ always holds; 
otherwise $\tilde{\pvec} = \zeros_\nk$, which is an infeasible solution. 
Therefore, $[-\Ib_\U \ \ib_\nk]$ always has rank $|\U| + 1$. 
From the positive definiteness of $\nabla_\pvec^2 \Omega_k(\tilde{\pvec})$, we have 
\[
\det\left(
\nabla_\xvec H(\tilde{\xvec}, \gvec_k)
\right)
=
\det\left(
\nabla_\pvec^2 \Omega_k(\tilde{\pvec})
\right)
\det\left(
- [-\Ib_\U \ \ib_\nk]^\top  \nabla_\pvec^2 \Omega_k(\tilde{\pvec})^{-1} [-\Ib_\U \ \ib_\nk]
\right) 
\neq0, 
\]
where we used the Schur complement. 
Hence $\nabla_\xvec H(\tilde{\xvec}, \gvec_k)$ is non-singular.  
This guarantees that $\nabla_{\gvec_k} \pvec_k$ can be computed by using the implicit function theorem as follows (see, e.g., \citep{dontchev2014implicit}): 
\begin{align}
\begin{bmatrix}
\nabla_{\gvec_k} \pvec_k 
\\
\nabla_{\gvec_k} \tilde{\lambdavec}_\U 
\\
\nabla_{\gvec_k} \tilde{\mu}
\end{bmatrix}
=
- \nabla_\xvec H(\tilde{\xvec}, \gvec_k)^{-1} \nabla_{\gvec_k} H(\tilde{\xvec}, \gvec_k). 
\end{align}
Thus, once the KKT triplet, the Hessian, and $\nabla_\thetavec \gvec_k(\thetavec)$ are obtained, we can compute 
$\nabla_\thetavec \pveck{k}(\thetavec)
=
\nabla_{\gvec_k} \pveck{k}(\gvec_k) \cdot \nabla_\thetavec \gvec_k(\thetavec)$ 
in $\mathrm{O}(\nk^3 + \nk^2\times \dim\Theta)$ time in general. 
For speeding up this step, we can reduce the $\nk$ value by using the stochastic version of the greedy algorithm \citep{mirzasoleiman2015lazier} (see, \Cref{subsec:stoc_gradient}). 

A recent result \citep{stechlinski2018generalized} 
provides 
an extended version of the implicit function theorem, 
which may enable us to deal with a wider class of $\Omega_k$; 
we leave this for future work.

\section{Discussion on other gradient estimators}\label{sec:other_grad}
The score-function gradient estimator is one of 
major Monte Carlo gradient estimators.  
Other than that, 
the pathwise and measure-valued gradient estimators are 
widely used (see, \citep{mohamed2019monte} for a survey). 
The Gumbel-Softmax estimator \citep{jang2017categorical,maddison2017concrete} has also been used 
in many recent studies. 
We discuss why it is difficult to use those estimators for our case. 

The pathwise gradient estimators basically use derivatives of quantities 
inside the expectation. 
In our case, however, we cannot differentiate the quantity, $Q(S)$, w.r.t. $S$ 
since the domain is non-continuous.

The measure-valued gradient estimators require us to 
decompose $\nabla_\thetavec p(\thetavec) $ 
into $p^+(\thetavec)$ and $p^-(\thetavec)$, 
which must satisfy the following conditions: 
both $p^+(\thetavec)$ and $p^-(\thetavec)$ 
form some probability distribution functions, 
and 
$\nabla_\thetavec p(\thetavec) 
= 
c_\thetavec (p^+(\thetavec) - p^-(\thetavec))$ holds with some constant $c_\thetavec$. 
Once we obtain a decomposition satisfying these conditions, 
we can estimate the gradient by sampling from $p^+(\thetavec)$ and $p^-(\thetavec)$. 
It is known that we can obtain such a decomposition 
when $p(\thetavec)$ has certain structures, e.g., 
Poisson and Gaussian. 
In our case, however, $p(\thetavec)$ is the output distribution, and how to decompose it is non-trivial; in fact, this seems to be very difficult. 

The Gumbel-Softmax estimator is obtained by continuously interpolating discrete categorical distributions (defined on $\Delta^{\nk}$ in our case) and computing derivatives at interior points. 
In our case, however, 
we must obtain an extreme point, $\Sk{k-1}$, in the ($k-1$)-th step 
to compute the categorical distribution $\pvec_k(\thetavec)$ used in the 
$k$-th step. That is, unlike the cases of \citep{jang2017categorical,maddison2017concrete}, 
{\sgreedy} 
sequentially samples from categorical distributions that depend on the past samples. 
Consequently, 
the continuous interpolation for a single step does not 
work for smoothing the sequential $\argmax$; hence we cannot apply the Gumbel-Softmax estimator to our setting.

\section{Learning submodular models with limited oracle queries}\label{sec:other_application}

We discuss the application of our framework to learning of parameterized submodular functions with limited oracle queries.   
We also provide experiments on learning deep submodular functions \citep{dolhansky2016deep}.

\subsection{Problem description}\label{subsec:blackbox}
We consider maximizing unknown submodular function $\hat \f(\cdot)$ by sequentially querying its values. 
Specifically, in each $t$-th round, 
we can query $\hat \f(\cdot)$ values at $N$ points $S_1,\dots, S_N\in \Ical$, 
and by using this feedback, 
we seek a good solution for maximizing $\hat \f(\cdot)$.  
We suppose that no prior knowledge on the true function, $\hat \f(\cdot)$, other than the fact that it is normalized, monotone, and submodular, is available and 
that to query the true function value is costly and time-consuming. 
We want to achieve high $\hat\f(\cdot)$ values with a small number of rounds and queries.  
One can think of this setting as a variant of 
submodular maximization with low adaptive complexities \citep{balkanski2018adaptive} 
or 
online submodular maximization with bandit feedback \citep{zhang2019online}. 

We consider the following approach: 
we construct some parameterized submodular model 
$\f(\cdot, \thetavec)$, e.g., a deep submodular function, 
and update $\thetavec$ by using our gradient estimators 
with $Q(S_j) = \hat{\f}(S_j)$. 
That is, 
akin to the decision-focused approach described in \Cref{subsec:decision}, 
we train $\f(\cdot, \thetavec)$ so that the greedy algorithm can achieve high $\hat \f(\cdot)$ values; 
the current setting is more difficult since we 
know nothing about $\hat \f(\cdot)$ in advance and 
features, which are used by the predictive models, are unavailable. 

\subsection{Experiments}\label{subsec:exp_blackbox}

\begin{figure}[tb]
	\centering
	\begin{minipage}[t]{.4\textwidth}
		\includegraphics[width=1.0\textwidth]{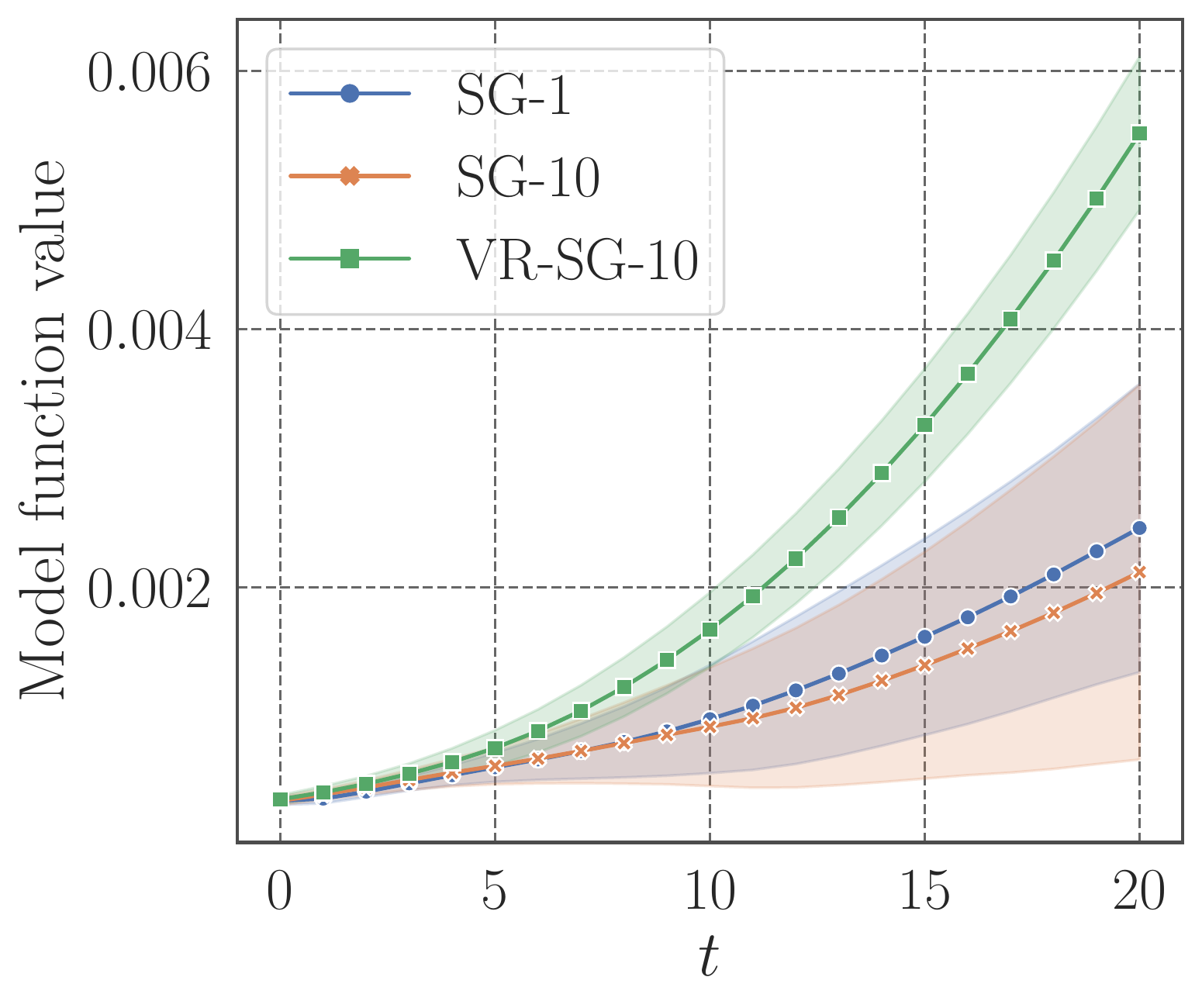}
		\subcaption{Noise-free}
		\label{fig:model}
	\end{minipage}
	\begin{minipage}[t]{.4\textwidth}
		\includegraphics[width=1.0\textwidth]{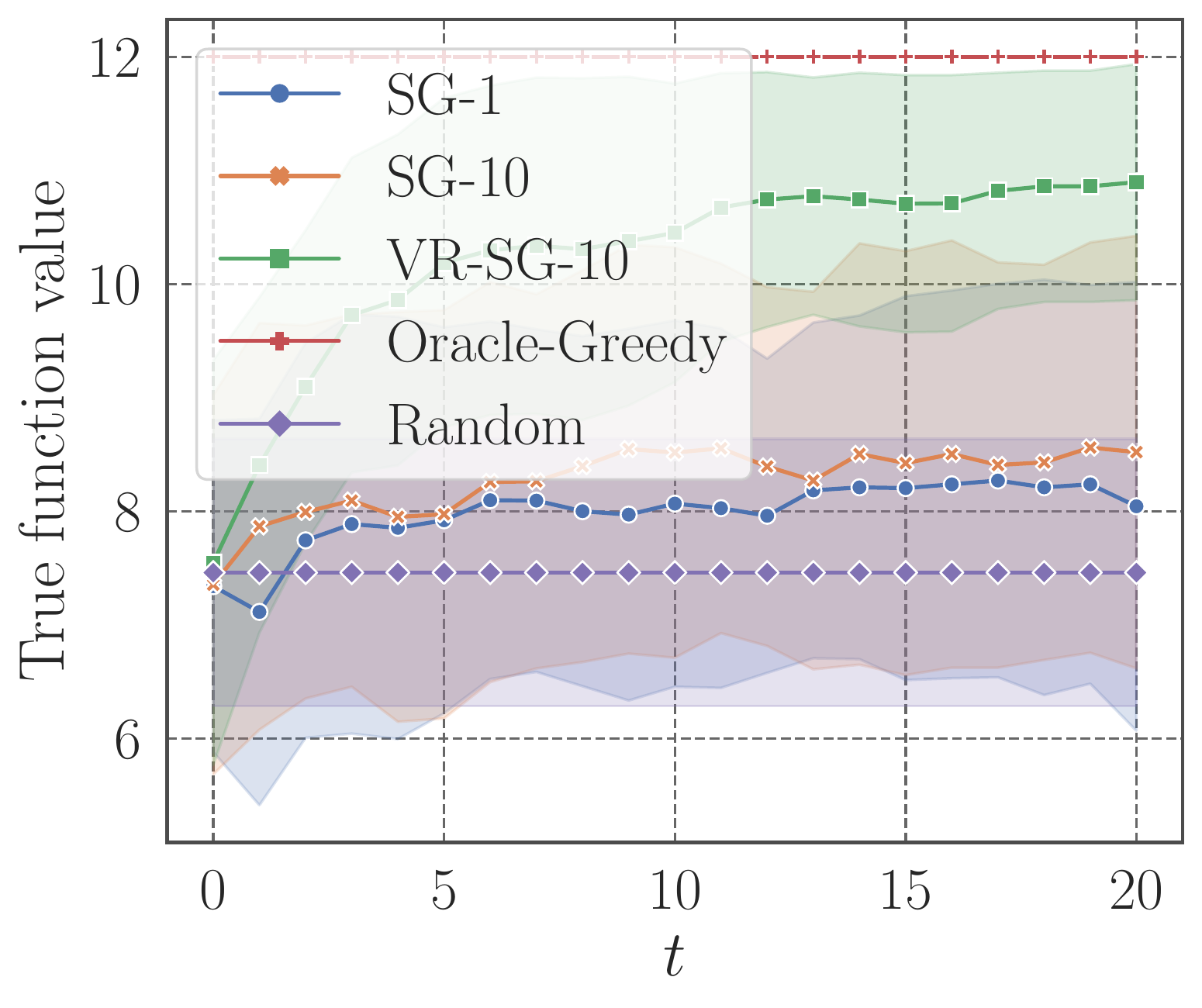}
		\subcaption{Noise-free}
		\label{fig:true}
	\end{minipage}
	\begin{minipage}[t]{.4\textwidth}
		\includegraphics[width=1.0\textwidth]{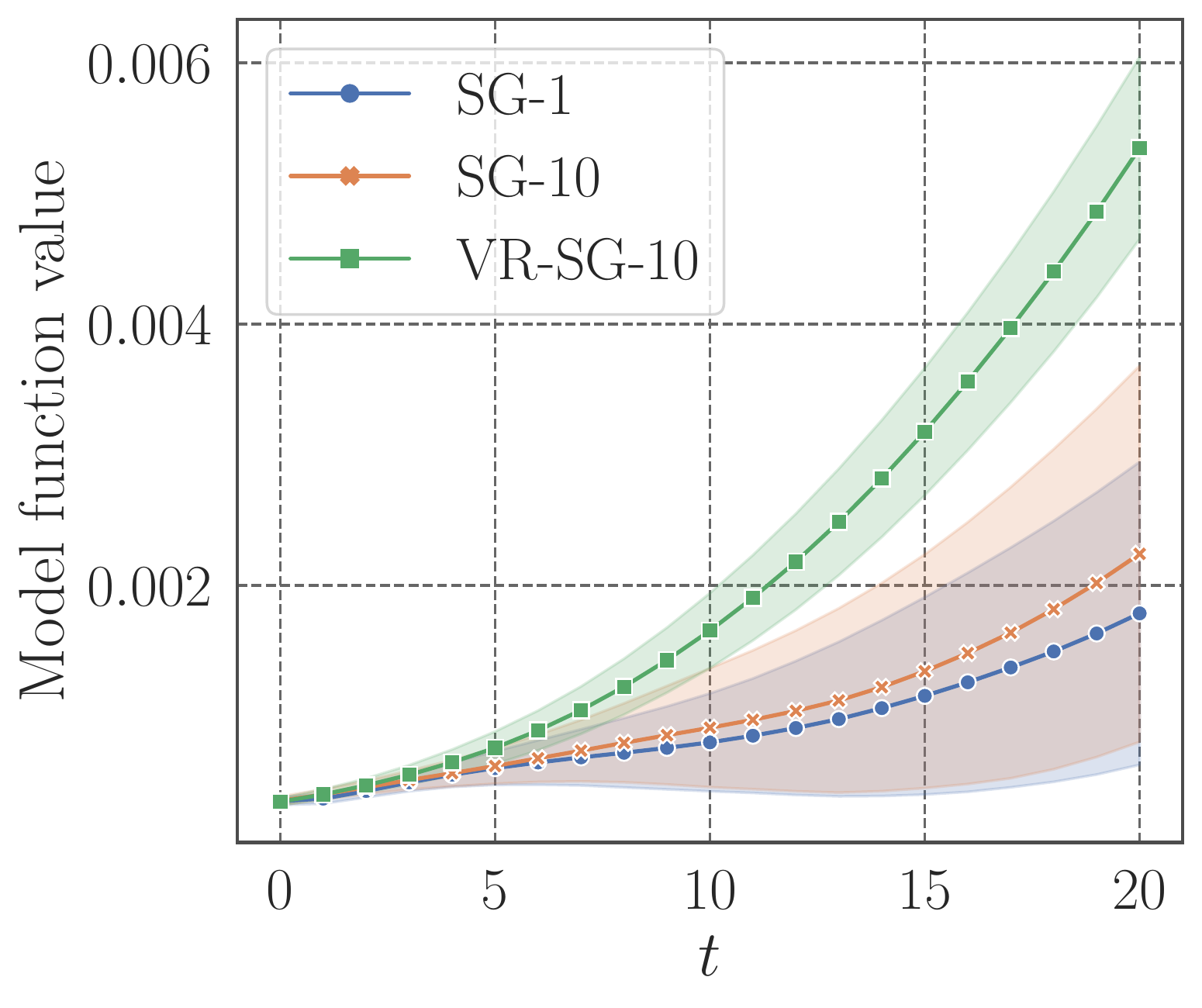}
		\subcaption{Noisy}
		\label{fig:noisy_model}
	\end{minipage}
	\begin{minipage}[t]{.4\textwidth}
		\includegraphics[width=1.0\textwidth]{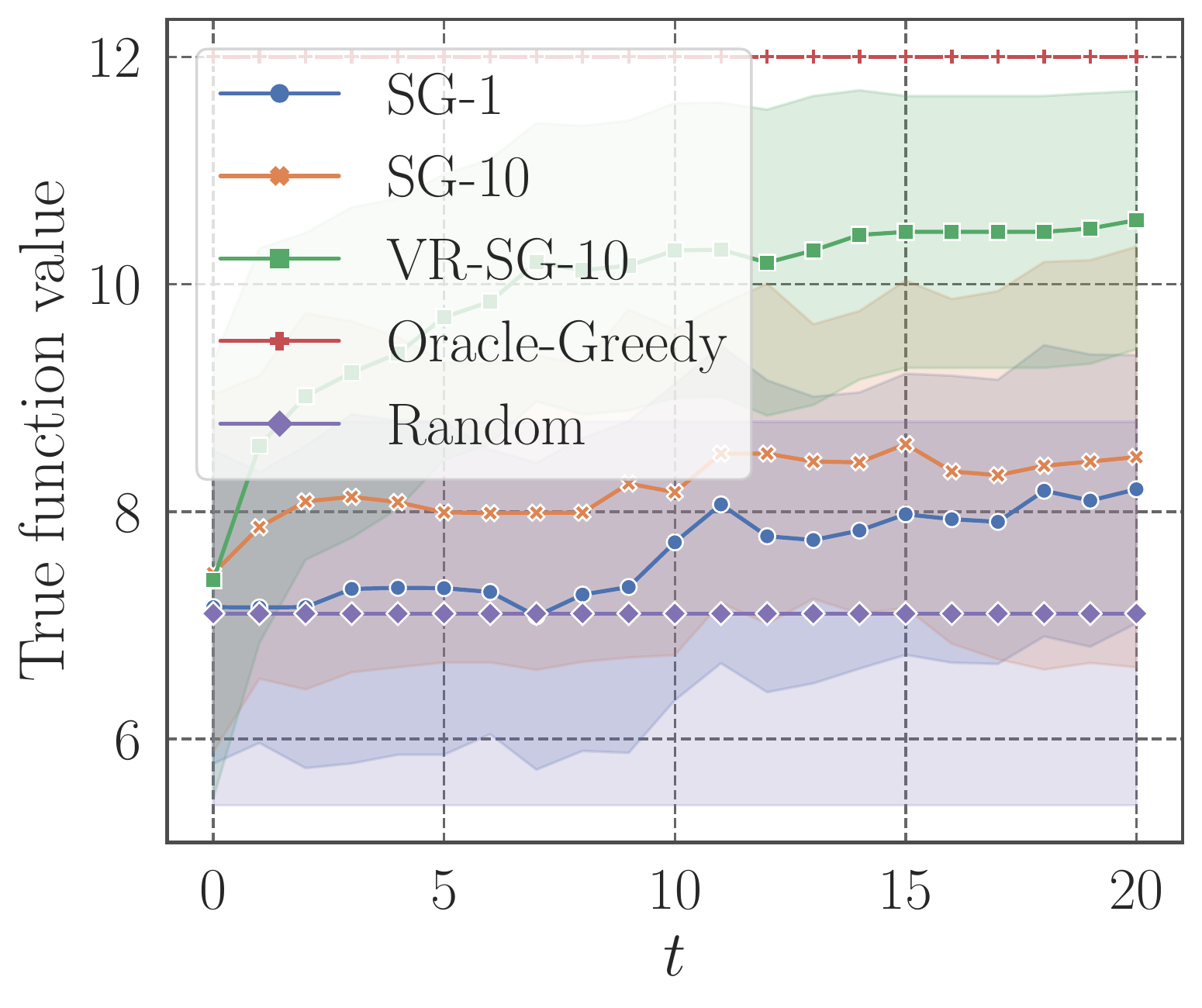}
		\subcaption{Noisy}
		\label{fig:noisy_true}
	\end{minipage}
	\caption{
		Mode and true function values  for noise-free (upper) and noisy (lower) settings.		 
	}
	\label{fig:dsf}
\end{figure}

We consider a situation where we make contact with business leaders to make influences on their companies. 
We use the corporate leadership network dataset of KONECT \citep{barnes2010structual, kunegis2013konect}, 
which contains person--company leadership information between $20$ people and $24$ companies; 
the companies are indexed with $i=1,\dots,24$. 
We let each $v\in V$ represent a person, who is associated with a subset of companies $I_v\subseteq \{1,\dots,24\}$.  
We define $I_X\coloneqq \bigcup_{v\in X} I_v$ for every $X\subseteq V$.  
We express the importance of the $i$-th company with a non-negative weight $w_i$; 
we let $w_1=w_3=\dots=w_{23}=1$ and $w_2=w_4=\cdots=w_{24}=0.1$.  
We use the weighted coverage function as an unknown true function: 
$\hat \f(X)\coloneqq \sum_{i \in I_X} w_i$.   
We separate the $20$ people into two groups of $10$ people, 
and we choose up to two people from 
each group; i.e., $(V, \Ical)$ forms a partition matroid.   

As a model function, $\f(\cdot, \thetavec)$, we use a deep submodular function that forms a $2$-layer NN. 
We set the hidden-layer size at $50$ and use sigmoid activation functions. 
We set initial NN parameters $\thetavec$ at non-negative values drawn uniformly at random from $[0, 0.01]$. 

For $t=1,\dots,20$, 
we perform {\sgreedy} $N$ times with objective function $\f(\cdot,\thetavec)$; 
we thus obtain $S_1,\dots, S_N\in\Ical$. 
We then query $\hat \f(S_1),\dots, \hat \f(S_N)$ values, 
with which we compute the gradient estimator, 
and we update $\thetavec$ by using Adam with learning rate $10^{-3}$.  
In each $t$-th round, we evaluate the quality of the trained model, $\f(\cdot,\thetavec)$, as follows: 
we obtain $S\in\Ical$ by applying the greedy algorithm to $\f(\cdot,\thetavec)$, 
and compute model function value $\f(S,\thetavec)$ 
and true function value $\hat\f(S)$. 
We also consider a noisy setting where observed $\hat \f(S_j)$ values are perturbed with random variables drawn from the standard normal distribution. 
We here use the entropy function with $\epsilon=0.02$ as a regularization function of {\sgreedy}.

As in \Cref{subsec:exp_decision}, 
({\bf VR-}){\bf SG-$N$} stands for (variance-reduced) {\sgreedy} with $N$ samples.   
We compare the true function values of ({\bf VR-}){\bf SG-$N$} with those of two methods: {\bf Oracle-Greedy}  and {\bf Random}.  
{\bf Oracle-Greedy} 
is the greedy algorithm directly applied to $\hat \f(\cdot)$, which we assume to be unknown in this setting; 
we use {\bf Oracle-Greedy} to see what if we had full access to the unknown true $\hat \f(\cdot)$. 
{\bf Random} returns $X\in\Ical$ by randomly choosing two people from each of the two groups. 

\Cref{fig:dsf} presents the means and standard deviations of the model and true function values over $30$ runs. 
We see that, by updating the model function, we can increase the true function values.  
As indicated by the results of {\bf SG-$1$}, 
if even once we can query $\hat\f(\cdot)$ value in each round, 
we can do better than {\bf Random}. 
With more queries and the variance reduction method, we can achieve higher true function values. 
The results suggest that our method is useful for learning and maximizing submodular functions when very limited prior knowledge and feedback are available.

\newcommand{\ssgreedy}{{\sc Stochastic Smoothed Greedy}}
\newcommand{\epsssg}{\varepsilon}
\newcommand{\nksg}{{n_k}}

\section{Differentiable stochastic greedy algorithm}\label{sec:stochastic_greedy}
We show that our framework can be used for making the stochastic greedy algorithm \citep{mirzasoleiman2015lazier} differentiable,  
which is a faster randomized variant of the greedy algorithm.  
In this section, we focus on the cardinality constrained case.

\Cref{alg:ssgreedy} presents the smoothed version of the stochastic greedy algorithm, which we call {\ssgreedy}.  
The only difference from {\sgreedy} (\Cref{alg:pgreedy}) 
is in Step \ref{stoc_step:uk}, 
where we sample $\nksg$ elements uniformly at random without replacement 
from $V\backslash S$ . 
In what follows, we let $\nk = \ceil{\frac{n}{K} \ln \frac{1}{\epsssg} }$ 
for every $k=1,\dots, K$, 
where $\epsssg \in (0, 1)$ is a hyper-parameter;  
this plays a role of controlling the speed--accuracy trade-off.

As with the original stochastic greedy algorithm, 
{\ssgreedy} requires only $\Orm(n \ln {1}/{\epsssg})$ evaluations 
of $\f(\cdot, \thetavec)$, 
while {\sgreedy} requires $\Orm(nK)$. 
Moreover,  
as explained in \Cref{subsec:stoc_gradient},
the gradient estimator for {\ssgreedy} 
can be computed more efficiently than that for {\sgreedy}. 
Therefore, {\ssgreedy} 
is useful when $n$ and $K$ are large and/or 
the evaluation of $\f(\cdot, \thetavec)$ is costly. 
Below we prove the approximation guarantee of {\ssgreedy}, 
and explain how to compute gradient estimators. 
We also present experiments to see the empirical speed--accuracy trade-off. 

\begin{algorithm}[tb]
	\caption{\ssgreedy}
	\label{alg:ssgreedy}
	\begin{algorithmic}[1]
		\State $S \gets \emptyset$
		\For{$k = 1,2\dots, K$}
		\State $\Uk = \{u_1,\dots,u_{\nk}\} \gets \text{ $\nksg$ elements chosen from $V\backslash S$ uniformly at random}$ 
		\label{stoc_step:uk}
		\State $\gveck{k}(\thetavec) = 
		(g_k(u_1, \thetavec),\dots, g_k(u_\nk, \thetavec))
		\gets (\fdelp{u_1}{S}{\thetavec},\dots,\fdelp{u_{\nk}}{S}{\thetavec})$
		\State $\pveck{k}(\thetavec) = 
		(\pk{k}(u_1, \thetavec),\dots, \pk{k}(u_\nk, \thetavec))
		\gets \argmax_{\pvec\in\Delta^{\nk}}\{\iprod{\gveck{k}(\thetavec), \pvec} - \Omega_k(\pvec)   \}$
		\State $s_k \gets u\in \Uk$ with probability $\pk{k}(u, \thetavec)$
		\State $S\gets S\cup\{s_k\}$ 
%		\If{$S$ is maximal} \Return $S$
%		\EndIf
		\EndFor
		\Return $S$        
	\end{algorithmic}
\end{algorithm}

\subsection{Approximation guarantee}\label{subsec:stoc_approximation_guarantee}

We prove that \Cref{alg:ssgreedy} returns solution $S$ 
that satisfies the following approximation guarantee for the cardinality constrained case. 

\begin{restatable}{thm}{cardinalitysg}\label{thm:ssgreedy} 
	If $\nksg \ge \frac{n}{K} \ln \frac{1}{\epsssg}$ for $k=1,\dots,K$, 	
	 we have 
	$\E[\f(S, \thetavec)] \ge (1-1/\e-\epsssg) \f(O,\thetavec) - \delta K$. 
\end{restatable}

\begin{proof}
	As with the proofs in \Cref{asec:proofs},
	we omit $\thetavec$ and use $S_k$ to denote the solution 
	obtained in the $k$-th step ($k=0,\dots, K$). 
	We take all random quantities to be conditioned on the realization of 
	the $(k-1)$-th step. 
	Once $\Uk$ is fixed in Step \ref{stoc_step:uk}, 
	we can obtain the following inequality from \Cref{lem:fdel}: 
	\begin{align}
	\E[\fdel{s_k}{S_{k-1}} \relmid{} \Uk] 
	\ge \fdel{s^*_k}{S_{k-1}} - \delta, 
	\end{align}
	where 
	$s^*_k \in \argmax_{u \in \Uk} \fdel{u}{S_{k-1}}$ 
	and
	$\E[\cdot\relmid{} \Uk]$ 
	denotes the expectation conditioned on $\Uk$. 
	By taking the expectation over all possible choices of $\Uk$, we obtain 
	\begin{align}\label{eq:fdelsg}
	\E[\f(S_k)] - \f(S_{k-1})
	=
	\E[\fdel{s_k}{S_{k-1}}] \ge \E[\fdel{s^*_k}{S_{k-1}}] - \delta.  
	\end{align}
	Note that here, $s^*_k$ is a random variable representing an element, 
	which the original stochastic greedy algorithm adds to the current solution. 	
	As proved in \citep{mirzasoleiman2015lazier}, 
	if $\nksg \ge \frac{n}{K} \ln \frac{1}{\epsssg}$, we have 
	\begin{align}\label{eq:s_marginal}
	\E[\fdel{s^*_k}{S_{k-1}}] \ge \frac{1-\epsssg}{K} ( \f(O \cup S_{k-1}) -\f(S_{k-1}) ).
	\end{align}	
	By substituting this inequality into \eqref{eq:fdelsg} and taking the expectation over all 
	possible realizations of $S_{k-1}$, we obtain 
	\begin{align}\label{eq:ss_marginal}
	\E[\f(S_k)] - \E[\f(S_{k-1})] \ge \frac{1-\epsssg}{K} ( \E[\f(O \cup S_{k-1})] -\E[\f(S_{k-1})] ) - \delta,    
	\end{align}	
	which holds for $k=1,\dots, K$. 
	Therefore, by induction, we obtain the theorem as follows: 
	\[ 
	\E[\f(S)] \ge \left( 1 - \left(1 - \frac{1-\epsssg}{K}\right)^K \right) \f(O) 
	- \delta\sum_{k=0}^{K-1}\left(1 - \frac{1-\epsssg}{K}\right)^k
	\ge \left(1 - \frac{1}{\e} - \epsssg \right) \f(O) - \delta K,  
	\]
	where we used 
	$\E[\f(O \cup S_{k-1})] \ge \f(O)$, 
	$\f(\emptyset)=0$,  
	and $\E[\f(S)] = \E[\f(S_K)]$. 	
\end{proof}

%\memo{to do; mirz lem for non-monotone}
%To deal with the non-monotonicity, 
%we perform Step \ref{stoc_step:uk} as follows: 
%we sample $\nksg$ elements from $V\backslash S$, 
%compute marginal gain $\fdel{u}{S}$ for each sampled element $u$,  
%and add $u$ to $\Uk$ if the $\fdel{u}{S} > 0$. 
%...this step seems to be non-trivial
%
%\begin{proof}
%	As in \citep[Lemma 1]{sakaue2019guarantees}, inequality \eqref{eq:s_marginal} holds even in the non-monotone case. 
%	
%	\begin{align}
%	\E[\f(S_{k})] \ge
%	\\
%	\E[\f(S_{k-1})] + \frac{1-\epsssg}{K} ( \f(O \cup S_{k-1}) -\f(S_{k-1}) ) - \delta
%	\\
%	\left( 1 - \frac{1-\epsssg}{K}\right) \E[\f(S_{k-1})] 
%	+ 
%	\frac{1-\epsssg}{K} \left( 1 - \frac{1}{K} \ln \frac{1}{\epsssg} - \frac{2}{n-K} \right)^{k-1} 
%	\f(O) - \delta
%	\end{align}
%\end{proof}

\subsection{Gradient estimation}\label{subsec:stoc_gradient}

We show how to compute gradient estimators for {\ssgreedy}. 
As with the case of {\sgreedy}, 
outputs of {\ssgreedy} are distributed over $\Sscr_{\le K}$. 
Therefore, the score-function gradient estimator can be computed by sampling outputs as in \Cref{sec:gradient}, i.e.,  
\begin{align}
\nabla_\thetavec \E_{S\sim p(\thetavec)}[\Q(S)] 
\approx
\frac{1}{N} \sum_{j = 1}^{N} \Q(S_j) \nabla_\thetavec \ln p(S_j,\thetavec)  
\quad \text{where} \quad 
S_j = (s_{1}, \dots, s_{|S_j|}) \sim p(\thetavec). 
\end{align}
Note that here $p(\thetavec)$ denotes the output distribution of 
the {\ssgreedy}; 
more precisely, 
for $\pvec_1(\thetavec), \dots, \pvec_K(\thetavec)$ 
and solution $S = \{s_1,\dots, s_K\}$ computed by \Cref{alg:ssgreedy}, 
we let 
$p(S, \thetavec) = \prod_{k=1}^{|S|} \pk{k}(s_k, \thetavec)$, 
where $\pk{k}(s_k, \thetavec)$ is the entry of $\pvec_k(\thetavec)$
corresponding to $s_k\in \Uk$. 
We can compute $\nabla_\thetavec \ln p(S_j,\thetavec)$ 
in the same manner as in \Cref{sec:gradient}.
%, and thus
%we can obtain unbiased gradient estimators for {\ssgreedy}. 

%We then see how the computation cost is reduced. 
Remember that the computation of $\nabla_\thetavec \ln p(S_j,\thetavec)$ 
involves the following differentiation based on the chain rule: 
$\nabla_\thetavec \pveck{k}(\thetavec)
=
\nabla_{\gvec_k} \pveck{k}(\gvec_k) \cdot \nabla_\thetavec \gvec_k(\thetavec)$. 
Here, the dimensionality of $\pvec_k$ and $\gvec_k$ is at most 
$\nk = \ceil{\frac{n}{K} \ln \frac{1}{\epsssg} }$, 
while it is up to $n$ in the case of {\sgreedy}. 
Therefore, {\ssgreedy} is effective for speeding up 
the computation of gradient estimators. 

\begin{figure}[tb]
	\centering
	\begin{minipage}[t]{.3\textwidth}
		\includegraphics[width=1.0\textwidth]{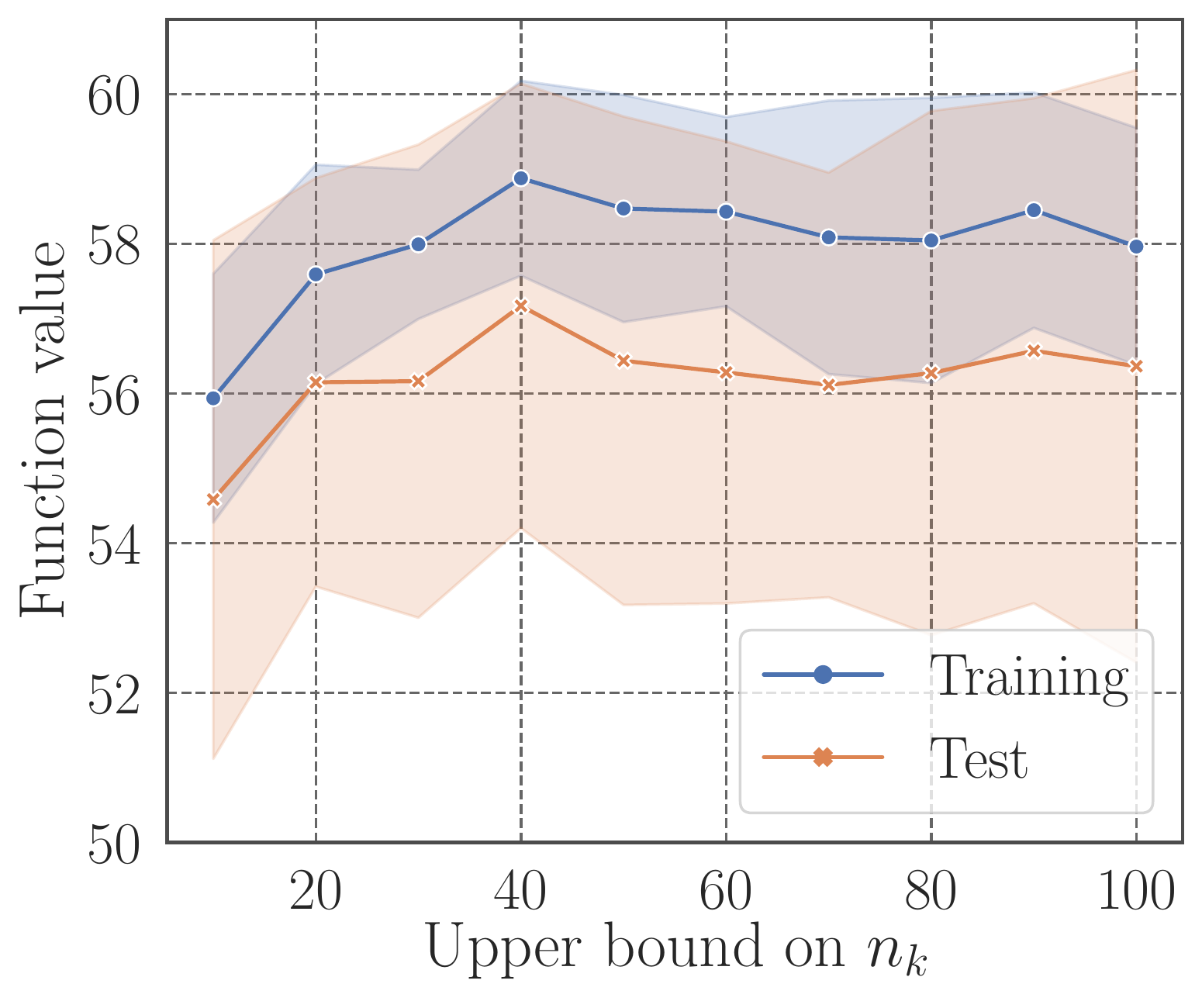}
		\subcaption{Training and test scores}
		\label{fig:stoc_fvals}
	\end{minipage}
	\begin{minipage}[t]{.3\textwidth}
		\includegraphics[width=1.0\textwidth]{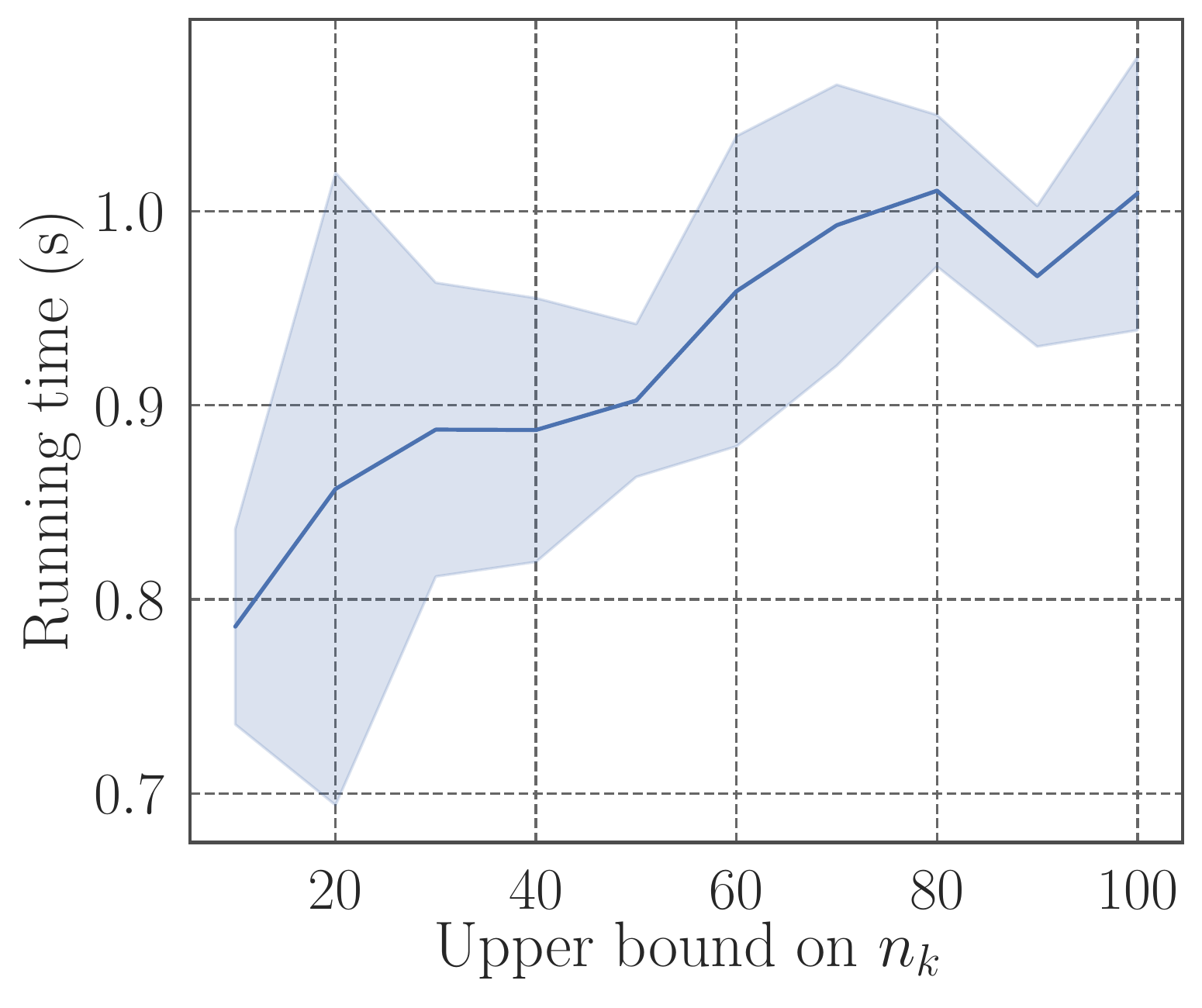}
		\subcaption{Solution computation}
		\label{fig:stoc_sgrd}
	\end{minipage}
	\begin{minipage}[t]{.3\textwidth}
		\includegraphics[width=1.0\textwidth]{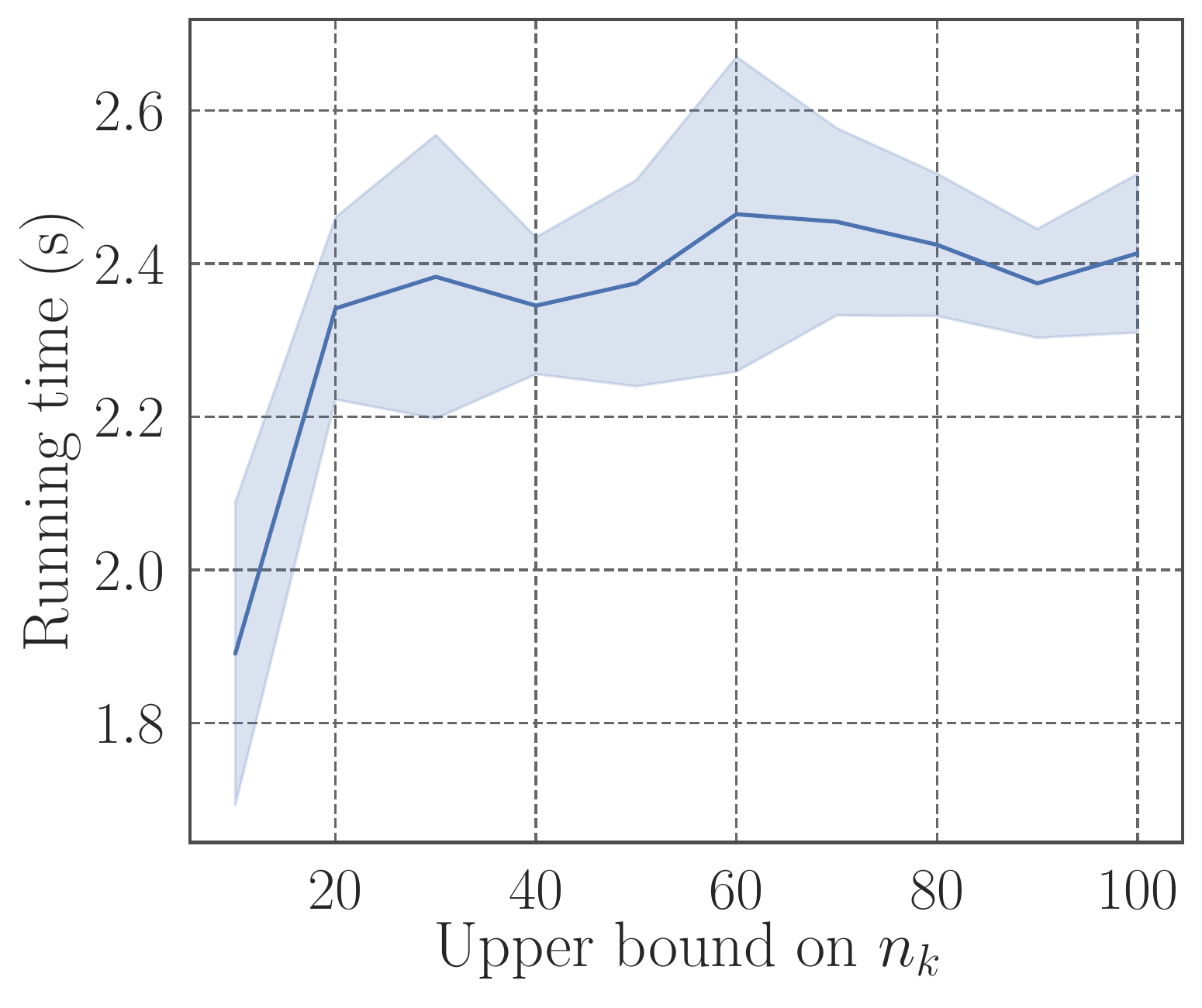}
		\subcaption{Gradient estimation}
		\label{fig:stoc_grad}
	\end{minipage}
	\caption{
		(a): 
		Objective function values achieved for the training and test instances.  
		(b): 
		Running times of {\ssgreedy} for updating $\wvec$ once (summation over $200$ runs).
		(c): 
		Running times of the gradient-estimator computation.
	}
	\label{fig:stoc}
\end{figure}

\subsection{Experiments}
We study the empirical performance of {\ssgreedy}. 
We use the same settings as those of the decision-focused learning experiments with $K=10$ (see, \Cref{subsec:exp_decision}), 
where we have $n=100$. 
We apply the {\ssgreedy} version of {\bf VR-SG-$10$} to the instances. 
We consider various upper-bound values, 
$10, 20, \dots, 100$, 
on $\nksg$; that is, 
in Step~\ref{stoc_step:uk} of \Cref{alg:ssgreedy}, 
we set $\nksg$ at the upper-bound value if it is less than $|V\backslash S|$ and at $|V\backslash S|$ otherwise. 

We evaluate objective function values with training and test instances for each upper bound on $\nksg$, 
where 
we calculate the means and standard deviations over $30$ training/test splits as in \Cref{subsec:exp_decision}. 
We also observe running times required for computing solutions with {\ssgreedy} 
and estimating gradients. 
More precisely,    
we measure those times taken for once updating the predictive-model parameter, $\wvec$; 
since the mini-batch size is $20$ and we perform $N=10$ trials, 
we take the sum of times over $200$ runs as the running time of {\ssgreedy}. 
In this experiment, 
$\wvec$ is updated $600$ times in total; we have $80 / 20 = 4$ mini-batches for each of $5$ epochs, and we consider $30$ random training/test splits, hence $4 \times 5 \times 30 = 600$. The running times of {\ssgreedy} and gradient estimation will be indicated with means and standard deviations over the $600$ iterations.

%\Cref{fig:stoc} presents the results.  
As shown in \Cref{fig:stoc_fvals}, 
even if $\nksg$ decreases, the objective function values do not drop so much with both training and test instances; rather, the highest values are achieved with $\nk=40$.  
The results imply that the stochastic greedy algorithm remains empirically effective even if it is smoothed with our framework. 
\Cref{fig:stoc_sgrd,fig:stoc_grad} 
confirm that by decreasing $\nksg$, we can reduce the running times required for 
computing solutions and estimating gradients. 
In this experimental setting, since the instance size is not so large and objective function values can be efficiently computed via matrix-vector products, 
the run-time overhead becomes dominant; 
this makes the degree of the speed-up yielded by decreasing $\nksg$ appears less significant.  
However, when instance sizes are larger and evaluations of objective functions are more costly, the speed-up achieved by using {\ssgreedy}  becomes more significant.

\end{document}